\documentclass{obnotes}

\newcommand{\sgn}{\mathop{}\mathopen{}\mathrm{sgn}}
\DeclareMathOperator{\image}{im}

\begin{document}

\title{Twisted Injectivity in PEPS and the Classification of Quantum Phases}
\author{Oliver Buerschaper}

\maketitle

\begin{abstract}
    We introduce a class of projected entangled pair states
    (\textsc{PEPS}) which is based on a group symmetry
    twisted by a 3-cocycle of the group. This twisted symmetry
    gives rise to a new standard form for \textsc{PEPS} from which
    we construct a family of local Hamiltonians which are gapped,
    frustration-free and include fixed points of the renormalization group flow.
    Moreover, we advance the classification of 2D gapped quantum spin
    systems by showing how this new standard form for \textsc{PEPS} determines
    the emergent topological order of these local Hamiltonians. Specifically,
    we identify their universality class as
    \textsc{Dijkgraaf-Witten} topological quantum field theory
    (\textsc{TQFT}).
\end{abstract}

\section{Introduction}
\subsection{Background}

A central topic in condensed matter physics is understanding
the universality classes of Hamiltonians with an energy gap above
their ground state(s). Here two
physical systems are defined to be in the same universality class
if they can be connected by a smooth path of
gapped local Hamiltonians (i.e. by quasi-adiabatic evolution)~\cite{Chen:2010gb,Bachmann:2011kw}.

Recently, in 1D a complete classification of gapped quantum spin systems
has been obtained using matrix product states
(\textsc{MPS})~\cite{Chen:2011iq} and
their parent Hamiltonians~\cite{Schuch:2011p3031}. This remarkable achievement
has largely been possible for the following reasons: there is an
area law for the entanglement entropy of gapped quantum spin
chains~\cite{Hastings:2007bu,Arad:2013tl}, there is a normal
form for \textsc{MPS}~\cite{PerezGarcia:2007wt}, and most importantly,
there is \emph{no intrinsic} topological order in 1D, only symmetry protected topological order.
In other words, in the absence of (physical) symmetry all ground states of
gapped quantum spin chains are equivalent to product states
under quasi-adiabatic evolution.
 
In 2D the situation turns out to be much more intricate.
First and foremost, there is \emph{both}
intrinsic~\cite{Tsui:1982kq,Laughlin:1983hk,Kitaev:2006ik,Levin:2004p117,Yan:2011kt,Zhang:2012jc,Parameswaran:2013uf}
\emph{and} symmetry protected topological
order~\cite{Schnyder:2009tg,Kitaev:2009vc,Hasan:2010ku,Qi:2011hb,Chen:2011p3277},
as well as nontrivial blends of these kinds of order~\cite{Mesaros:2013es,Hung:2012ud}.
Second, there is no area law for gapped quantum spin systems in 2D (or higher)
which follows from the gap alone.
Third, proving the very existence of a gap is notoriously
hard in 2D.
Finally, tensor network methods generalizing \textsc{MPS} to 2D, like
projected entangled pair states (\textsc{PEPS}), turn out to be considerably
more challenging than in 1D, if only because contracting \textsc{PEPS}
is not efficient in general~\cite{Schuch:2007wq} and a condition known as injectivity
does not automatically imply a gap for the parent
Hamiltonian of a \textsc{PEPS}.
As far as understanding universality classes is concerned the lack of a normal form for \textsc{PEPS}
(or for any other local description of ground states) arguably poses the most formidable challenge.  

Much more can be said once we restrict our attention to a remarkable
class of 2D quantum spin systems in which the gap is provably stable against \emph{any} small
enough perturbation~\cite{Bravyi:2010jn,Michalakis:2011p3299}. In fact, this unconditional stability
under quasi-adiabatic evolution
may be taken as a definition of intrinsic topological order.
It then follows that a whole range of low energy features,
both \emph{local} and \emph{global}, will also survive
the perturbation, i.e. are characteristic of the universality class.
Unfortunately, locally computable invariants of the ground
states~\cite{Bravyi:2010jn,Kitaev:2006,Levin:2006ij} reveal only partial
information about the kind of topological order~\cite{Flammia:2009p2617}.
Not surprisingly, global invariants~\cite{Hastings:2004cd,Hastings:2005cs}
are potentially much more valuable for identifying the universality class
of a system.
Along these lines, representations of the modular group~\cite{Kitaev:2006ik} are
believed to completely describe intrinsic topological order,
and compelling numerical evidence~\cite{Zhang:2012jc,Cincio:2013ku,Zaletel:2013fw,Zhu:2013wt}
suggests that these are
indeed stable.
However, it is far from clear how these global invariants \emph{emerge} from
the microscopic interactions, or from a yet to be discovered
\emph{local} normal form for low energy states.
As far as \textsc{PEPS} are concerned, a first step towards solving this problem has been made
by studying \textsc{PEPS} with a virtual group symmetry and
a certain standard form derived from it~\cite{Schuch:2010p2806}. 

In this article we generalize this approach and present a) a large subclass of \textsc{PEPS}
which is based on a new standard form and b) a continuous family of 
corresponding parent Hamiltonians which are gapped, frustration-free
and include fixed points of the renormalization group flow.
We advance
the classification of 2D gapped quantum spin systems by
showing how this new standard form for \textsc{PEPS}
determines the emergent topological order in the system.

\subsection{Results}

More precisely, our new standard form for \textsc{PEPS}
fits into the existing picture as follows. Decompose
any \textsc{PEPS} tensor~$A$ as $A=QW$ where $Q$ is positive and
$W$ is a partial isometry, called the isometric
form in~\cite{Schuch:2011p3031}. A major part of classifying
all \textsc{PEPS} consists in identifying all possible isometric (or standard) forms~$W$
in order to eventually obtain a complete normal form for \textsc{PEPS}.
The only standard forms~$W$ identified so far are
the \emph{block-diagonal} and the \emph{$G$-isometric} form
where $G$ is a finite group~\cite{Schuch:2011p3031}. While the former standard form describes
some kind of local symmetry breaking, the latter describes the
universality class of quantum double models $\mathrm{D}(G)$~\cite{Kitaev:2003}.
Here we extend the virtual group symmetry underlying the $G$-isometric standard
form~\cite{Schuch:2010p2806} to a \emph{twisted} virtual group symmetry where
the twist~$\omega$ is a 3-cocycle of the group~$G$. As our first main result
we obtain the subclass of $(G,\omega)$-injective tensors together with
the new $(G,\omega)$-isometric standard form for \textsc{PEPS}.

Furthermore we employ this new standard form for \textsc{PEPS}
to construct a continuous family of gapped, frustration-free Hamiltonians
whose interaction terms do \emph{not} commute. Nevertheless we
can show rigorously that this family
lies in the universality class of twisted gauge theory (twisted quantum
double models)~\cite{Dijkgraaf:1990p2241,deWildPropitius:1997p3081,Hung:2012vw,Hu:mIFulyz6}.
We proceed in three steps. First, we show that
the parent Hamiltonian of an arbitrary $(G,\omega)$-injective \textsc{PEPS}
on a torus has a ground state degeneracy which is given
by the number of so called $c^\omega$-regular
pair conjugacy classes of~$G$. Second, we show that parent Hamiltonians
of $(G,\omega)$-isometric \textsc{PEPS} describe
time slices in \textsc{Dijkgraaf-Witten} topological quantum
field theory (\textsc{TQFT})~\cite{Dijkgraaf:1990p2241}, only depend on the cohomology
class of the twist~$\omega$, and have commuting projections as
interaction terms. This clearly identifies the above universality
class. Finally, we use quasi-adiabatic evolution to lift all
interesting features of the fixed point Hamiltonians of the previous
step to the so called \emph{almost}
$(G,\omega)$-isometric \textsc{PEPS}. As our second main result,
this procedure yields the desired family of gapped, frustration-free
Hamiltonians whose intrinsic topological order is determined
by our new standard form for \textsc{PEPS}.

\subsection{Structure}

This article is structured as follows. In Section~\ref{sec:twisted_injectivity}
we introduce $(G,\omega)$-injective tensors via a virtual symmetry
expressed as matrix product operators (\textsc{MPO}).
We also study the space of quantum states that
emerges from arbitrary $(G,\omega)$-injective tensors on a torus within the usual
\textsc{PEPS} formalism.
In
Section~\ref{sec:hamiltonians} we construct the parent
Hamiltonians corresponding to $(G,\omega)$-injective tensors.
In Section~\ref{sec:twisted_isometry} we identify our new $(G,\omega)$-isometric standard
form for \textsc{PEPS} and study the properties of their parent Hamiltonians. In
Section~\ref{sec:almost_twisted_isometry} we turn to
\emph{almost} $(G,\omega)$-isometric tensors and obtain the desired
family of gapped, frustration-free Hamiltonians.
We conclude with a discussion
and an outlook on open questions in Section~\ref{sec:discussion}.

\section{Twisted Injectivity}
\label{sec:twisted_injectivity}

\subsection{Branched Tensors}

\begin{definition}[Branched polygon]
    A \emph{branching structure} on a polygon is an acyclic
    orientation of its edges.
\end{definition}

\begin{definition}[Branched tensor]
    Let $G$ be a finite group and $p$ be an $n$-polygon with a branching
    structure.
    A \emph{branched tensor}~$(A_i)$ is a
    tensor with elements~$(A_i)_{\alpha_1\dots\alpha_n}\in\mathbb{C}$ for $\alpha_k\in G$
    and $i=1,\dots,N$ together with the branching structure of~$p$.
\end{definition}

For reasons to be become clear later we will usually duplicate each index~$\alpha_k$.
For example, this gives a particular branched triangle tensor:
    \begin{equation}
        \label{eq:branched_triangle_tensor}
        \vcorrect{34bp}{\includegraphics[scale=0.25]{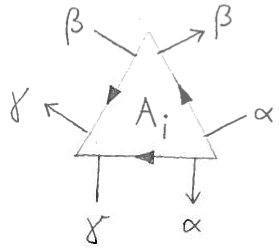}}
        \coloneqq
         (A_i)_{\alpha_1
                \alpha_2
                \alpha_3}.
    \end{equation}

Given a branched tensor~$(A_i)_{\alpha_1\dots\alpha_n}$, we may naturally associate
with it the linear map
\begin{equation*}
    A
    \coloneqq
     \sum_i
     \sum_{\alpha_k\in
           G}
     (A_i)_{\alpha_1
            \dots
            \alpha_n}\mkern2mu
     \ketbra{i}
            {\alpha_1,
             \dots,
             \alpha_n}
\end{equation*}
from the \emph{virtual} space~$\mathbb{C}G\otimes\dots\otimes\mathbb{C}G$
to the local \emph{physical} space~$\mathcal{H}_p$ (of the polygon~$p$).
We will use
\begin{equation*}
    \mathcal{L}_p
    \coloneqq
     \image(A)
\end{equation*}
to denote the physical subspace generated from the virtual level.%
\footnote{While we may always choose~$\mathcal{H}_p=\mathcal{L}_p$
    to make $A$ surjective locally we will have
    $\mathcal{L}_{p\cup p'}\subsetneq\mathcal{H}_{p\cup p'}$
    once we contract tensors~$(A_i)$ and~$(B_j)$ corresponding
    to~$p$ and~$p'$.}
Frequently we
will also consider branched tensors~$(\tilde{A}^i)_{\alpha_1\dots\alpha_n}$ in order to
describe linear maps from the local physical space to the virtual space:
\begin{equation*}
    \tilde{A}
    \coloneqq
     \sum_i
     \sum_{\alpha_k\in
           G}
     (\tilde{A}^i)_{\alpha_1
                    \dots
                    \alpha_n}\mkern2mu
     \ketbra{\alpha_1,
             \dots,
             \alpha_n}
            {i}.
\end{equation*}
As is customary, we will usually not
distinguish between the tensors and their corresponding maps.

\subsection{Twisted Symmetry as Matrix Product Operator}

We first define a twisted group action on the virtual boundary of 
a branched tensor in terms of an \textsc{MPO}.

\begin{definition}[Twisted symmetry \textsc{MPO}]
    Let $\omega\in H^3\bigl(G,\mathrm{U}(1)\bigr)$.
    For each $g\in G$ the tensors~$T_+^\omega(g)$ and~$T_-^\omega(g)$
    whose only non-vanishing elements read
    \begin{align}
        \vcorrect{28bp}{\includegraphics[scale=0.25]{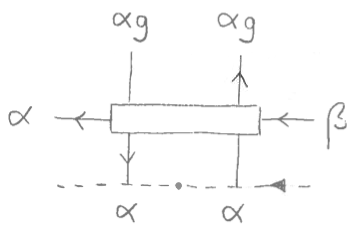}}\mkern3mu
        & \coloneqq
           \omega(\alpha
                  \beta^{-1}\mkern-3mu,
                  \beta,
                  g)^{-1}, \\
        \vcorrect{30bp}{\includegraphics[scale=0.25]{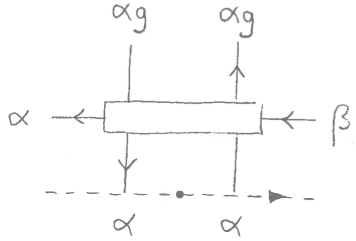}}\mkern3mu
        & \coloneqq
           \omega(\beta
                  \alpha^{-1}\mkern-3mu,
                  \alpha,
                  g),
    \end{align}
    respectively
    define an \textsc{MPO}~$V^\omega(g)$ acting on the virtual boundary
    of a branched tensor via
    \begin{equation}
        \expval{\alpha_1,
                \dots,
                \alpha_n}
               {V^\omega(g)}
               {\beta_1,
                \dots,
                \beta_n}
        \coloneqq
         \mkern3mu
         \vcorrect{38bp}{\includegraphics[scale=0.25]{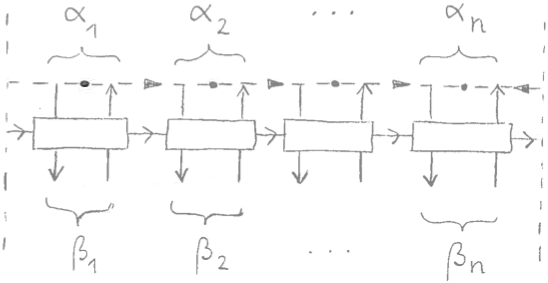}}\mkern6mu.
    \end{equation}
\end{definition}

\begin{remark}
    We can write the \textsc{MPO}~$V^\omega(g)$ in
    a slightly different way so that its \enquote{input}
    and \enquote{output} become more apparent:
    \begin{equation}
        \vcorrect{39bp}{\includegraphics[scale=0.25]{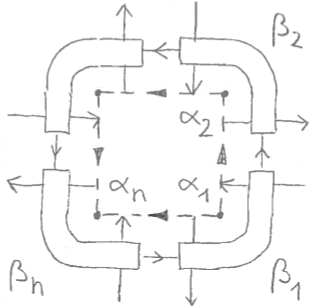}}\mkern3mu
        =\vcorrect{43bp}{\includegraphics[scale=0.25]{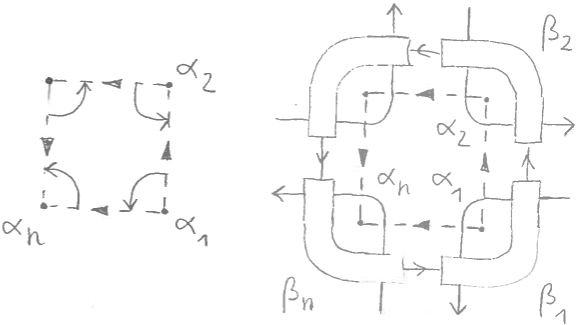}}\mkern6mu.
    \end{equation}
    This naturally extends to all other branched polygons beyond the one
    shown here.\qed
\end{remark}

\begin{proposition}
    The \textsc{MPO}s $\{V^\omega(g)\mid g\in G\}$ form a unitary
    representation of~$G$.
\end{proposition}
\begin{proof}
    We may assume that the \textsc{MPO}
    acts on the virtual boundary of the branched triangle tensor~\eqref{eq:branched_triangle_tensor}
    without loss of generality. Suppressing duplicate indices we have
    \begin{equation}
        V^\omega(g)
        =\smashoperator[l]{\sum_{\alpha,
                                   \beta,
                                   \gamma\in
                                   G}}
           \frac{\omega(\gamma
                        \alpha^{-1}\mkern-3mu,
                        \alpha
                        g^{-1}\mkern-3mu,
                        g)}
                {\omega(\beta
                        \alpha^{-1}\mkern-3mu,
                        \alpha
                        g^{-1}\mkern-3mu,
                        g)\mkern2mu
                 \omega(\gamma
                        \beta^{-1}\mkern-3mu,
                        \beta
                        g^{-1}\mkern-3mu,
                        g)}\mkern2mu
           \ketbra{\alpha
                   g^{-1}\mkern-3mu,
                   \beta
                   g^{-1}\mkern-3mu,
                   \gamma
                   g^{-1}}
                  {\alpha,
                   \beta,
                   \gamma}.
    \end{equation}
    
    From the 3-cocycle relation~\eqref{eq:3-cocycle_definition}
    it is not difficult to show that $V^\omega(g_1)\mkern2mu V^\omega(g_2)=V^\omega(g_1g_2)$
    for all~$\omega$ and $g_i\in G$.
    
    Finally we obtain
    \begin{align*}
        V^\omega(g^{-1})
        & =\smashoperator[l]{\sum_{\alpha,
                                   \beta,
                                   \gamma\in
                                   G}}
           \frac{\omega(\gamma
                        \alpha^{-1}\mkern-3mu,
                        \alpha
                        g,
                        g^{-1})}
                {\omega(\beta
                        \alpha^{-1}\mkern-3mu,
                        \alpha
                        g,
                        g^{-1})\mkern2mu
                 \omega(\gamma
                        \beta^{-1}\mkern-3mu,
                        \beta
                        g,
                        g^{-1})}\mkern2mu
           \ketbra{\alpha
                   g,
                   \beta
                   g,
                   \gamma
                   g}
                  {\alpha,
                   \beta,
                   \gamma} \\
        & =\smashoperator[l]{\sum_{\alpha,
                                   \beta,
                                   \gamma\in
                                   G}}
           \frac{\omega(\beta
                        \alpha^{-1}\mkern-3mu,
                        \alpha,
                        g)\mkern2mu
                 \omega(\gamma
                        \beta^{-1}\mkern-3mu,
                        \beta,
                        g)}
                {\omega(\gamma
                        \alpha^{-1}\mkern-3mu,
                        \alpha,
                        g)}\mkern2mu
           \frac{\omega(\gamma
                        \alpha^{-1}\mkern-3mu,
                        \alpha,
                        e)}
                {\omega(\beta
                        \alpha^{-1}\mkern-3mu,
                        \alpha,
                        e)\mkern2mu
                 \omega(\gamma
                        \beta^{-1}\mkern-3mu,
                        \beta,
                        e)} \\
        & \hphantom{=\smashoperator[l]{\sum_{\alpha,
                                             \beta,
                                             \gamma\in
                                             G}}{}}
           \ketbra{\alpha
                   g,
                   \beta
                   g,
                   \gamma
                   g}
                  {\alpha,
                   \beta,
                   \gamma} \\
        & =\smashoperator[l]{\sum_{\alpha,
                                   \beta,
                                   \gamma\in
                                   G}}
           \frac{\omega(\beta
                        \alpha^{-1}\mkern-3mu,
                        \alpha,
                        g)\mkern2mu
                 \omega(\gamma
                        \beta^{-1}\mkern-3mu,
                        \beta,
                        g)}
                {\omega(\gamma
                        \alpha^{-1}\mkern-3mu,
                        \alpha,
                        g)}\mkern2mu
           \ketbra{\alpha
                   g,
                   \beta
                   g,
                   \gamma
                   g}
                  {\alpha,
                   \beta,
                   \gamma} \\
        & =V^\omega(g)^\dagger
    \end{align*}
    for any~$\omega$ and $g\in G$. Here we repeatedly
    used~\eqref{eq:3-cocycle_definition} to derive
    \begin{equation*}
        \frac{\omega(\gamma
                     \alpha^{-1}\mkern-3mu,
                     \alpha,
                     e)}
             {\omega(\beta
                     \alpha^{-1}\mkern-3mu,
                     \alpha,
                     e)\mkern2mu
              \omega(\gamma
                     \beta^{-1}\mkern-3mu,
                     \beta,
                     e)}
        =\frac{1}
              {\omega(\beta,
                      e,
                      e)\mkern2mu
               \omega(\beta^{-1}\mkern-3mu,
                      \beta,
                      e)}
        =\frac{1}
              {\omega(e,
                      \beta,
                      e)}
        =1.
    \end{equation*}
    The last equality is indeed true because every 3-cocycle~$\omega$
    is equivalent to a \emph{normalized} 3-cocycle~$\tilde{\omega}$
    for some 2-cochain~$\phi$, hence:
    \begin{equation*}
        \omega(e,
               \beta,
               e)
        =\tilde{\omega}(e,
                        \beta,
                        e)\mkern2mu
         \frac{\phi(e,
                    \beta)\mkern2mu
               \phi(\beta,
                    e)}
              {\phi(e,
                    \beta)\mkern2mu
               \phi(\beta,
                    e)}
        =1.
    \end{equation*}
\end{proof}

\begin{corollary}
    \label{cor:twisted_projection}
    Let $n$ be the length of the virtual boundary.
    The map
    \begin{equation}
        \label{eq:projection}
        \mathcal{P}^\omega
        \coloneqq
         \frac{1}
              {\abs{G}}
         \sum_{g\in
               G}
         V^\omega(g)
    \end{equation}
    is a Hermitian projection onto the $(G,\omega)$-symmetric subspace
    with rank~$\abs{G}^{n-1}$.
\end{corollary}

\begin{definition}[Virtual symmetry]
    A branched tensor~$(A_i)$ is
    called $(G,\omega)$-symmetric if
    \begin{equation}
        \vcorrect{51bp}{\includegraphics[scale=0.25]{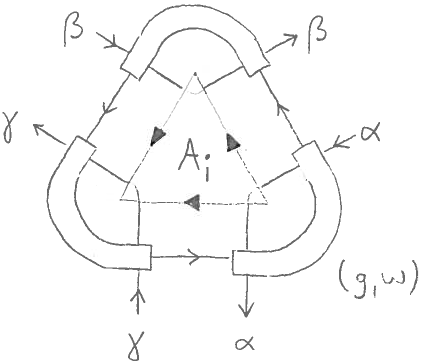}}
        =\mkern3mu
         \vcorrect{34bp}{\includegraphics[scale=0.25]{figure1}}
    \end{equation}
    for all $g\in G$, or equivalently, if
    \begin{equation}
        A
        \mathcal{P}^\omega
        =A
    \end{equation}
    where the underlying branching structure of~$\mathcal{P}^\omega$ is
    that of the virtual boundary of~$A$.
\end{definition}

We can contract branched tensors along common boundary edges (by
contracting the respective virtual indices), as long
as those edges have the same orientation.

\begin{lemma}[Concatenation of twisted symmetry]
    Any contraction $(P_{ij\dots})\coloneqq(A_iB_j\dots)$
    of compatible $(G,\omega)$-symmetric tensors~$(A_i)$,
    $(B_j)$, … is again $(G,\omega)$-symmetric
    and the twisted symmetry \textsc{MPO} only depends on the
    virtual boundary of~$P$.
\end{lemma}
\begin{proof}
    Let us first focus on two $(G,\omega)$-symmetric triangle
    tensors $(A_i)$ and $(B_j)$ contracted along a common oriented
    edge. Without loss of generality we assume
    \begin{equation}
        \label{eq:square_tensor}
        (P_{i
            j})
        =(A_i
          B_j)
        =\mkern3mu
         \vcorrect{37bp}{\includegraphics[scale=0.25]{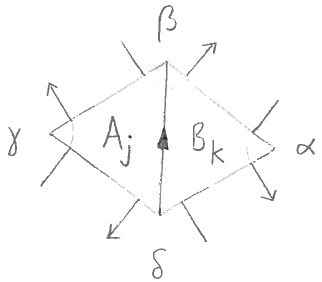}}
    \end{equation}
    and leave the orientation of the boundary edges implicit.
    Using the $(G,\omega)$-symmetry of the individual tensors
    and the identity
    \begin{equation}
        \label{eq:sym_rule_1}
        \vcorrect{52bp}{\includegraphics[scale=0.25]{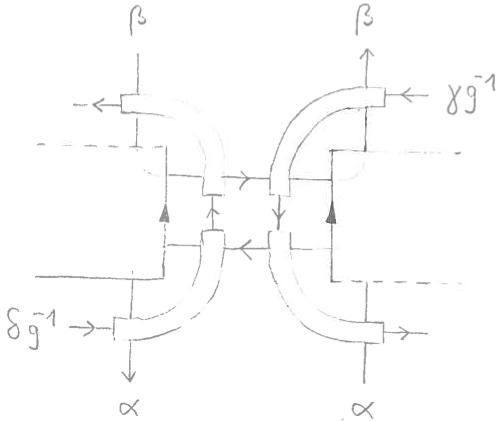}}
        =\vcorrect{52bp}{\includegraphics[scale=0.25]{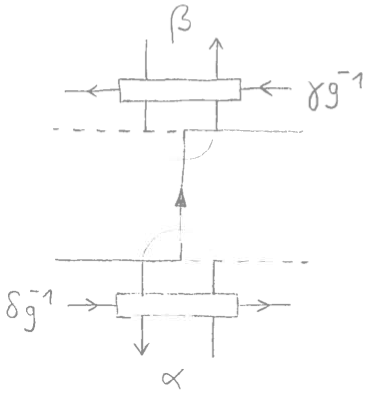}}
    \end{equation}
    we obtain
    \begin{equation*}
        (P_{i
            j})
        =\vcorrect{50bp}{\includegraphics[scale=0.25]{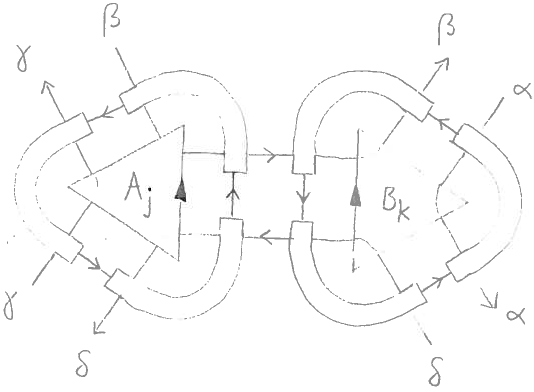}}
        =\vcorrect{47bp}{\includegraphics[scale=0.25]{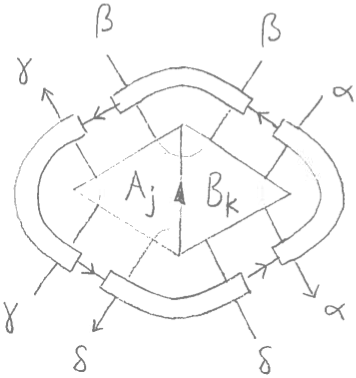}}
    \end{equation*}
    where all virtual symmetry tensors are of type~$(g,\omega)$.
    Note that \eqref{eq:sym_rule_1} equally holds for a downward
    interior edge and arbitrarily oriented boundary edges. This
    is because the only virtual symmetry tensors that depend on
    the orientation of the interior edge are the upper left and
    lower right ones. It is not difficult to see that these
    always cancel.
    
    Let us now show that the virtual symmetry only depends on \emph{open}
    indices at the virtual boundary. Suppressing irrelevant edge
    orientations we have the identity
    \begin{equation}
        \label{eq:sym_rule_2}
        \vcorrect{18bp}{\includegraphics[scale=0.25]{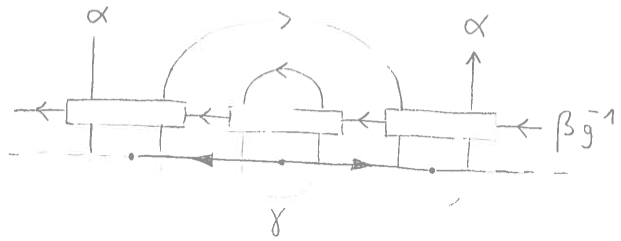}}
        =\mkern3mu
         \vcorrect{51bp}{\includegraphics[scale=0.25]{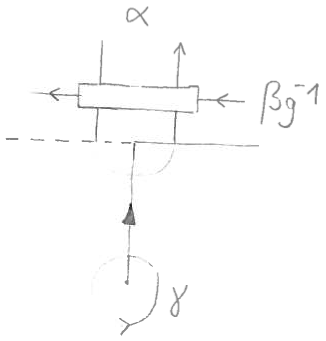}}
    \end{equation}
    because the two leftmost virtual symmetry tensors always cancel.
    It is immediately clear that this also holds for a downward interior
    edge as well as arbitrarily oriented boundary edges.
    
    Without loss of generality we may now assume that the contracted tensor~$(P_{ijk})$
    has the form
    \begin{equation}
        \vcorrect{35bp}{\includegraphics[scale=0.25]{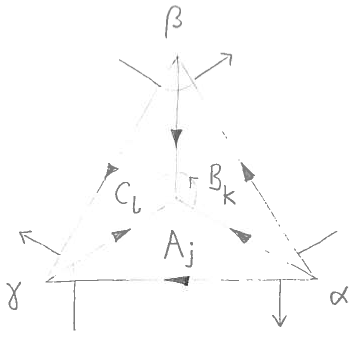}}
    \end{equation}
    After applying the $(G,\omega)$-symmetry to each triangle tensor
    individually we can push the virtual symmetry to the boundary by
    employing \eqref{eq:sym_rule_1} twice, followed by one application
    of~\eqref{eq:sym_rule_2}. This easily generalizes to arbitrary
    contracted tensors~$P$.
\end{proof}

\begin{definition}[Twisted injectivity]
    A $(G,\omega)$-symmetric tensor~$(A_i)$ is called
    $(G,\omega)$-injective if there exists a tensor~$(\tilde{A}^i)$
    such that
    \begin{equation}
        \label{eq:injectivity}
        \sum_i\mkern6mu
        \vcorrect{22bp}{\includegraphics[scale=0.25]{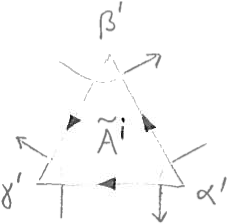}}\mkern12mu
        \vcorrect{22bp}{\includegraphics[scale=0.25]{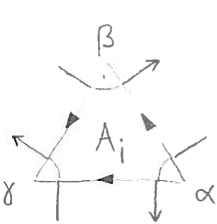}}
        =\expval{\alpha'\mkern-2mu,
                 \beta'\mkern-2mu,
                 \gamma'}
                {\mathcal{P}^\omega}
                {\alpha_1,
                 \alpha_2,
                 \alpha_3},
    \end{equation}
    or equivalently,
    \begin{equation}
        \tilde{A}
        A
        =\mathcal{P}^\omega
    \end{equation}
    where the underlying branching structure of both~$\tilde{A}$ and~$\mathcal{P}^\omega$
    is that of the virtual boundary of~$A$.
\end{definition}

\begin{lemma}[Concatenation of twisted injectivity]
    Any contraction $(P_{ij\dots})\coloneqq(A_iB_j\dots)$ of compatible
    $(G,\omega)$-injective tensors~$(A_i)$, $(B_j)$, … is again
    $(G,\omega)$-injective. 
\end{lemma}
\begin{proof}
    It suffices to prove the claim for two $(G,\omega)$-injective triangle
    tensor~$(A_i)$ and $(B_j)$ contracted along a common oriented edge
    as in~\eqref{eq:square_tensor}. Setting
    \begin{equation*}
        (\tilde{P}^{ij})_{\alpha_1
                          \dots
                          \alpha_4}
        \coloneqq
         \abs{G}\mkern4mu
         \vcorrect{33bp}{\includegraphics[scale=0.25]{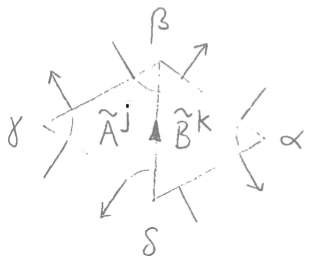}}
    \end{equation*}
    we obtain
    \begin{align*}
        \tilde{P}
        P
        & \simeq
           \abs{G}
           \sum_{j,
                 k}\mkern6mu
           \vcorrect{29bp}{\includegraphics[scale=0.25]{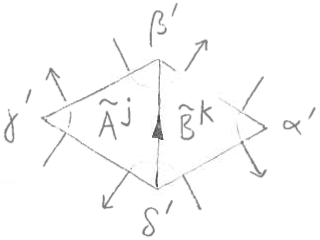}}\mkern9mu
           \vcorrect{35bp}{\includegraphics[scale=0.25]{figure9}}
           \displaybreak[0] \\
        & =\frac{1}
                {\abs{G}}
           \sum_{g,
                 h\in
                 G}\mkern3mu
           \vcorrect{35bp}{\includegraphics[scale=0.25]{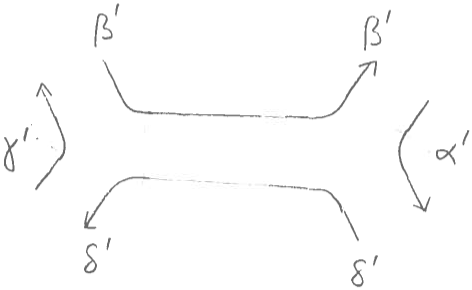}}\mkern12mu
           \vcorrect{49bp}{\includegraphics[scale=0.25]{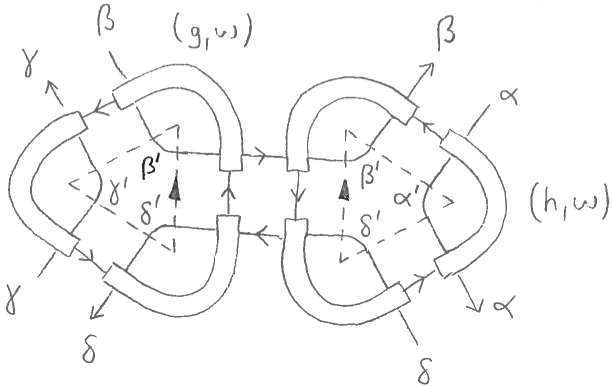}}
           \displaybreak[0] \\
        & =\frac{1}
                {\abs{G}}
           \sum_{g\in
                 G}\mkern3mu
           \vcorrect{48bp}{\includegraphics[scale=0.25]{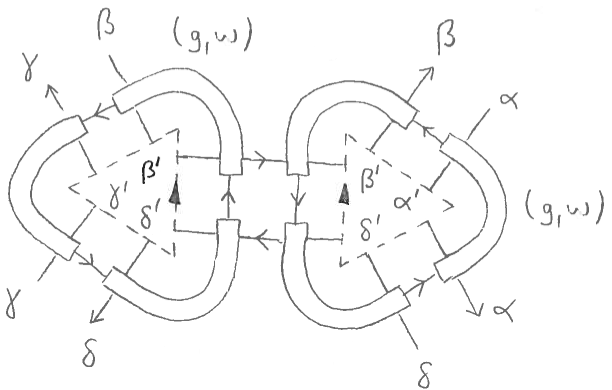}}
           \displaybreak[0] \\
        & =\frac{1}
                {\abs{G}}
           \sum_{g\in
                 G}\mkern3mu
           \vcorrect{48bp}{\includegraphics[scale=0.25]{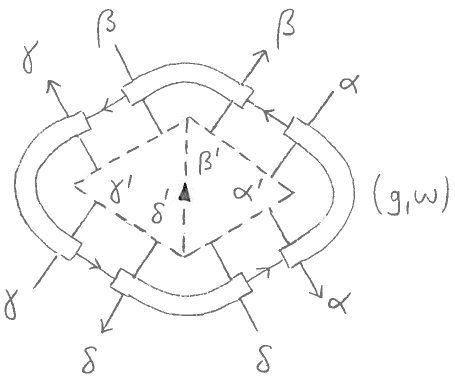}} \\
        & =\mathcal{P}^\omega
    \end{align*}
    where we used~\eqref{eq:injectivity} in the second step and~\eqref{eq:sym_rule_1} in the last one.
\end{proof}

\subsection{Closure on a Torus}

There are different ways to close a (minimal) torus and obtain a quantum state:
\begin{align}
    \ket{\psi(M)}
    & \coloneqq
       \sum_{j,
             k}
       \vcorrect{34bp}{\includegraphics[scale=0.25]{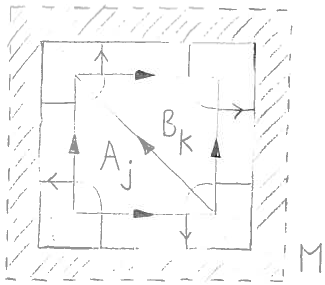}}
       \ket{j,
            k}, \\
    \ket{\psi(P)}
    & \coloneqq
       \sum_{j,
             k}
       \vcorrect{46bp}{\includegraphics[scale=0.25]{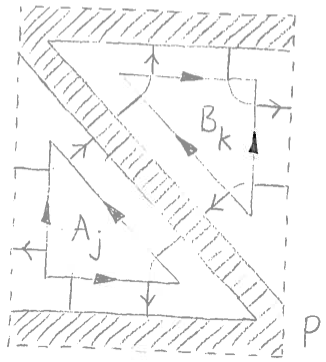}}
       \ket{j,
            k}, \\
    \ket{\psi(Q)}
    & \coloneqq
       \sum_{j,
             k}
       \vcorrect{36bp}{\includegraphics[scale=0.25]{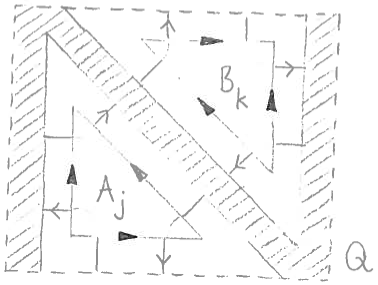}}
       \ket{j,
            k}.
\end{align}

\begin{figure}
    \centering
    \includegraphics[scale=0.25]{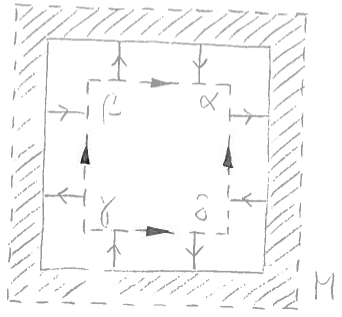}\hfil
    \includegraphics[scale=0.25]{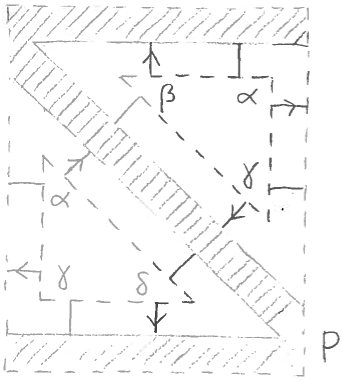}\hfil
    \includegraphics[scale=0.25]{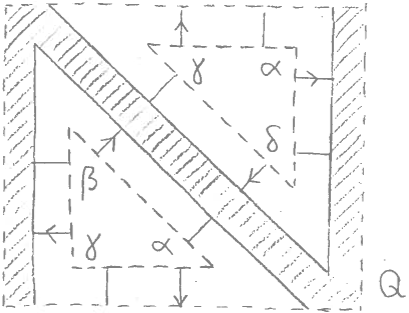}
    \caption
        [Closures]
        {The closure tensor elements~$\braket{\alpha_1,\dots,\alpha_4}{M}$,
         $\braket{\alpha_1,\dots,\alpha_4}{P}$ and
         $\braket{\alpha_1,\dots,\alpha_4}{Q}$.}
    \label{fig:closures}
\end{figure}

We are interested in states whose closure cannot be \enquote{detected} locally. In other
words, we would like to study those states which can be simultaneously expressed using \emph{any}
of the possible closures.

\begin{definition}
    For all $\alpha,g,h\in G$ with $hg=gh$ let $\lambda_X(\alpha;g,h)\in\mathbb{C}$.
    Define tensors $M(g,h)$, $P(g,h)$ and~$Q(g,h)$ by
    \begin{align}
        \braket{\alpha,
                \alpha
                g,
                \alpha
                g
                h,
                \alpha
                h}
               {M(g,
                  h)}
        & \coloneqq
           \lambda_M(\alpha;
                     g,
                     h),
           \displaybreak[0] \\
        \braket{\alpha,
                \alpha
                g,
                \alpha
                g
                h,
                \alpha
                h}
               {P(g,
                  h)}
        & \coloneqq
           \lambda_P(\alpha;
                     g,
                     h),
           \displaybreak[0] \\
        \braket{\alpha,
                \alpha
                g,
                \alpha
                g
                h,
                \alpha
                h}
               {Q(g,
                  h)}
        & \coloneqq
           \lambda_Q(\alpha;
                     g,
                     h),
    \end{align}
    all other elements being zero.
\end{definition}

\begin{remark}
    So far we have only specified which tensor elements
    do \emph{not} vanish, but left their actual value completely
    undetermined. In particular, we
    have not specified any relation between the three 
    functions $\lambda_X$. We will do so in due course.\qed
\end{remark}

\begin{lemma}
    \label{lem:closure_minimal_torus_span}
    \begin{multline}
        \label{eq:closure}
        \{\ket{\psi(M)}
          \mid
          M\}\cap
        \{\ket{\psi(P)}
          \mid
          P\}\cap
        \{\ket{\psi(Q)}
          \mid
          Q\} \\
        =\mathrm{span}\bigl\{\ket[\big]{\psi\bigl(M(g,
                                                    h)\bigr)}\bigm\vert
                             g,
                             h\in
                             G
                             \wedge
                             h
                             g
                             =g
                              h\bigr\}.
    \end{multline}
\end{lemma}
\begin{proof}
    Without loss of generality we study a minimal torus consisting
    of two $(G,\omega)$-injective triangle tensors~$(A_j)$ and~$(B_k)$.
    
    Let us first show that the right hand side is contained in the left
    hand side. Indeed, given the closure tensor~$M(g,h)$ we have
    \begin{equation*}
        \vcorrect{34bp}{\includegraphics[scale=0.25]{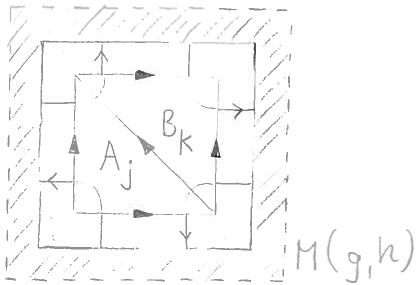}}
        =\vcorrect{50bp}{\includegraphics[scale=0.25]{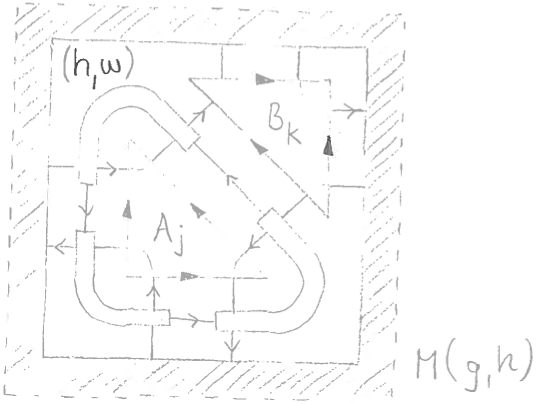}}
        =\vcorrect{36bp}{\includegraphics[scale=0.25]{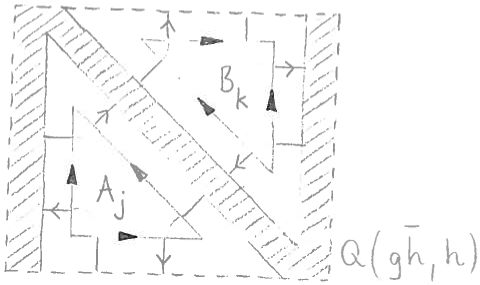}}
    \end{equation*}
    if we set
    \begin{equation*}
        \lambda_Q(\alpha;
                  g
                  h^{-1}\mkern-3mu,
                  h)
        \coloneqq
         \lambda_M(\alpha;
                   g,
                   h)\mkern2mu
         \frac{\omega\bigl((h^{-1})^\alpha\mkern-2mu,
                           \alpha
                           g,
                           h\bigr)}
              {\omega\bigl((g
                            h^{-1})^\alpha\mkern-2mu,
                           \alpha,
                           h\bigr)\mkern2mu
               \omega\bigl((g^{-1})^\alpha\mkern-2mu,
                           \alpha
                           g,
                           h\bigr)}.
    \end{equation*}
    Similarly one can show that
    $\ket[\big]{\psi\bigl(M(g,h)\bigr)}=\ket[\big]{\psi\bigl(P(g,g^{-1}h)\bigr)}$
    for some appropriately chosen~$\lambda_P(\alpha;g,h)$.
    
    \bigskip
    Now we want to show that the left hand side is contained in the right hand
    side. We may assume that all closure tensors~$M$, $P$,
    and~$Q$ are $(G,\omega)$-symmetric, i.e.
    $\mathcal{P}^\omega\mkern2mu\ket{X}=\ket{X}$,
    because the triangle tensors are. 
    Suppose $\ket{\psi(M)}=\ket{\psi(P)}=\ket{\psi(Q)}$
    for some closures~$M$, $P$ and~$Q$.
    We can access the closure~$M$ and express it in terms of~$P$ by applying
    a certain \enquote{inverse} map
    to $\ket{\psi(M)}=\ket{\psi(P)}$. Then we obtain
    \begin{align}
        \braket{\alpha,
                \beta,
                \gamma,
                \delta}
               {M}
        & =\frac{1}
                {\abs{G}}
           \sum_{g\in
                 G}
           \vcorrect{54bp}{\includegraphics[scale=0.25]{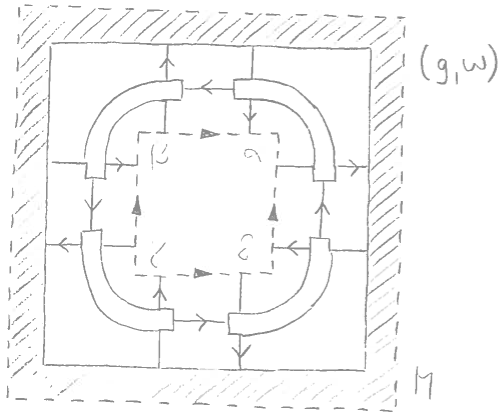}}
           \nonumber
           \displaybreak[0] \\
        & =\abs{G}
           \sum_{j,
                 k}
           \vcorrect{30bp}{\includegraphics[scale=0.25]{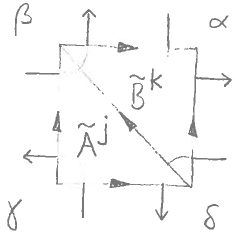}}\mkern12mu
           \vcorrect{34bp}{\includegraphics[scale=0.25]{figure24}}
           \nonumber
           \displaybreak[0] \\
        & =\abs{G}
           \sum_{j,
                 k}\mkern4mu
           \vcorrect{41bp}{\includegraphics[scale=0.25]{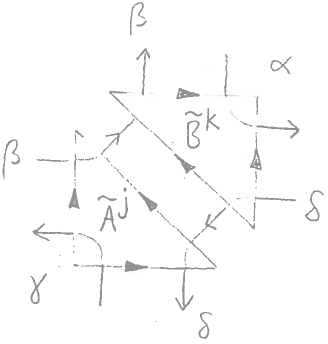}}\mkern12mu
           \vcorrect{46bp}{\includegraphics[scale=0.25]{figure25}}
           \nonumber
           \displaybreak[0] \\
        & =\frac{1}
                {\abs{G}}
           \sum_{g,
                 h\in
                 G}
           \delta_{\alpha
                   h,
                   \beta
                   g}\mkern2mu
           \delta_{\gamma
                   g,
                   \delta
                   h}\mkern4mu
           \vcorrect{75bp}{\includegraphics[scale=0.25]{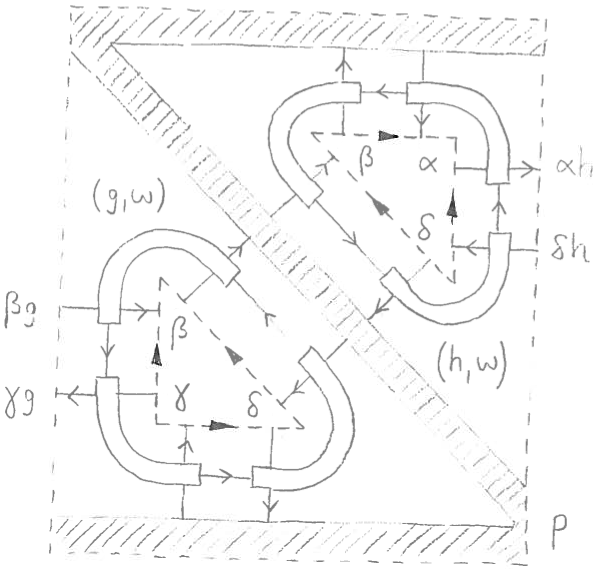}}
           \nonumber
           \displaybreak[0] \\
        & =\sum_{h\in
                 G}
           \delta_{\beta,
                   \alpha
                   h}\mkern2mu
           \delta_{\gamma,
                   \delta
                   h}\mkern2mu
           \frac{\omega(\alpha
                        \beta^{-1}\mkern-3mu,
                        \beta,
                        h)\mkern2mu
                 \omega(\beta
                        \delta^{-1}\mkern-3mu,
                        \delta,
                        h)}
                {\omega(\alpha
                        \delta^{-1}\mkern-3mu,
                        \delta,
                        h)}
           \nonumber \\
        & \hphantom{=\sum_{h\in
                           G}{}}
           \frac{1}
                {\abs{G}}
           \sum_{g\in
                 G}
           \frac{\omega(\alpha
                        \beta^{-1}\mkern-3mu,
                        \beta
                        h,
                        g)\mkern2mu
                 \omega(\beta
                        h
                        \gamma^{-1}\mkern-3mu,
                        \gamma,
                        g)}
                {\omega(\alpha
                        h
                        \delta^{-1}\mkern-3mu,
                        \delta,
                        g)\mkern2mu
                 \omega(\delta
                        \gamma^{-1}\mkern-3mu,
                        \gamma,
                        g)}\mkern2mu
           \braket{\alpha
                   h
                   g,
                   \beta
                   h
                   g,
                   \gamma
                   g,
                   \delta
                   g}
                  {P}
           \nonumber
           \displaybreak[0] \\
        & =\sum_{h\in
                 G}
           \delta_{\beta,
                   \alpha
                   h}\mkern2mu
           \delta_{\gamma,
                   \delta
                   h}\mkern2mu
           \frac{\omega(\alpha
                        \beta^{-1}\mkern-3mu,
                        \beta,
                        h)\mkern2mu
                 \omega(\beta
                        \delta^{-1}\mkern-3mu,
                        \delta,
                        h)}
                {\omega(\alpha
                        \delta^{-1}\mkern-3mu,
                        \delta,
                        h)}\mkern2mu
           \braket{\alpha
                   h,
                   \beta
                   h,
                   \gamma,
                   \delta}
                  {P}.
           \label{eq:closure_M_from_P}
    \end{align}
    
    We would like to express the closure~$P$ in terms of~$Q$ now. This
    can be achieved by applying a slightly different \enquote{inverse}
    map to $\ket{\psi(P)}=\ket{\psi(Q)}$:
    \begin{align}
        \braket{\alpha,
                \beta,
                \gamma,
                \delta}
               {P}
        & =\abs{G}
           \sum_{j,
                 k}\mkern3mu
           \vcorrect{48bp}{\includegraphics[scale=0.25]{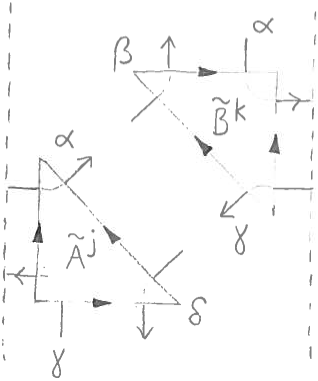}}\mkern6mu
           \vcorrect{46bp}{\includegraphics[scale=0.25]{figure25}}
           \nonumber
           \displaybreak[0] \\
        & =\abs{G}
           \sum_{j,
                 k}\mkern3mu
           \vcorrect{48bp}{\includegraphics[scale=0.25]{figure35}}\mkern6mu
           \vcorrect{36bp}{\includegraphics[scale=0.25]{figure26}}
           \nonumber
           \displaybreak[0] \\
        & =\frac{1}
                {\abs{G}}
           \sum_{g,
                 h\in
                 G}
           \delta_{\alpha
                   h,
                   \delta
                   g}\mkern2mu
           \delta_{\beta
                   h,
                   \gamma
                   g}
           \vcorrect{85bp}{\includegraphics[scale=0.25]{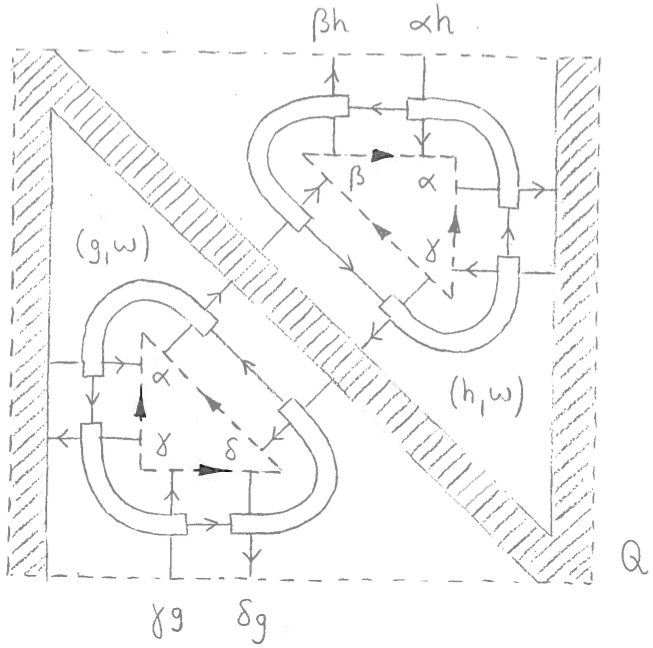}}
           \nonumber
           \displaybreak[0] \\
        & =\sum_{h\in
                 G}
           \delta_{\delta,
                   \alpha
                   h}\mkern2mu
           \delta_{\gamma,
                   \beta
                   h}\mkern2mu
           \frac{\omega(\alpha
                        \beta^{-1}\mkern-3mu,
                        \beta,
                        h)\mkern2mu
                 \omega(\beta
                        \gamma^{-1}\mkern-3mu,
                        \gamma,
                        h)}
                {\omega(\alpha
                        \gamma^{-1}\mkern-3mu,
                        \gamma,
                        h)}\mkern2mu
           \braket{\alpha
                   h,
                   \alpha,
                   \gamma,
                   \gamma
                   h}
                  {Q}.
           \label{eq:closure_P_from_Q}
    \end{align}
    
    Finally, we would like to express~$Q$ in terms of~$M$. As before, we
    can do this by applying a certain \enquote{inverse} map to
    $\ket{\psi(Q)}=\ket{\psi(M)}$:
    \begin{align}
        \braket{\alpha,
                \beta,
                \gamma,
                \delta}
               {Q}
        & =\abs{G}
           \sum_{j,
                 k}\mkern3mu
           \vcorrect{42bp}{\includegraphics[scale=0.25]{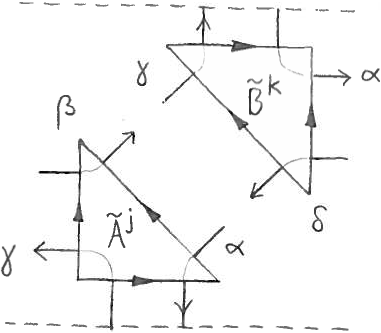}}\mkern6mu
           \vcorrect{34bp}{\includegraphics[scale=0.25]{figure24}}
           \nonumber
           \displaybreak[0] \\
        & =\frac{1}
                {\abs{G}}
           \sum_{g,
                 h\in
                 G}
           \delta_{\gamma
                   h,
                   \beta
                   g}\mkern2mu
           \delta_{\delta
                   h,
                   \alpha
                   g}\mkern4mu
           \vcorrect{62bp}{\includegraphics[scale=0.25]{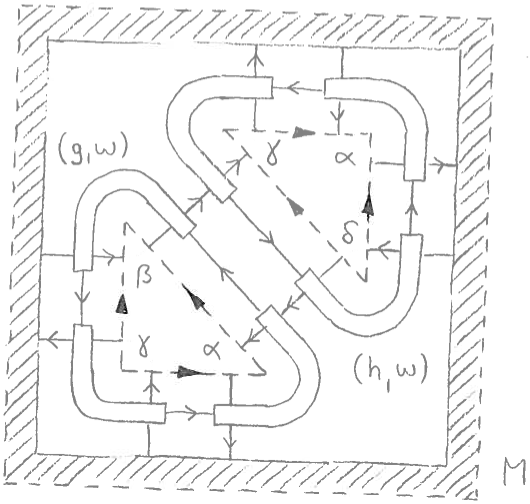}}
           \nonumber
           \displaybreak[0] \\
        & =\sum_{g\in
                 G}
           \delta_{\gamma,
                   \beta
                   g}\mkern2mu
           \delta_{\delta,
                   \alpha
                   g}\mkern2mu
           \frac{\omega(\beta
                        \gamma^{-1}\mkern-3mu,
                        \gamma,
                        g)}
                {\omega(\alpha
                        \gamma^{-1}\mkern-3mu,
                        \gamma,
                        g)\mkern2mu
                 \omega(\beta
                        \alpha^{-1}\mkern-3mu,
                        \alpha,
                        g)}\mkern2mu
           \braket{\alpha,
                   \beta
                   g,
                   \gamma
                   g,
                   \delta}
                  {M}.
           \label{eq:closure_Q_from_M}
    \end{align}
    
    Substituting~\eqref{eq:closure_P_from_Q} and~\eqref{eq:closure_Q_from_M}
    in~\eqref{eq:closure_M_from_P} we obtain
    \begin{equation}
        \label{eq:closure_M_final}
        \braket{\alpha,
                \beta,
                \gamma,
                \delta}
               {M}
        =\smashoperator[l]{\sum_{g,
                                 h\in
                                 G}}
         \delta_{h
                 g,
                 g
                 h}\mkern2mu
         \delta_{\beta,
                 \alpha
                 g}\mkern2mu
         \delta_{\gamma,
                 \alpha
                 g
                 h}\mkern2mu
         \delta_{\delta,
                 \alpha
                 h}\mkern2mu
         \braket{\alpha,
                 \alpha
                 g,
                 \alpha
                 g
                 h,
                 \alpha
                 h}
                {M}.
    \end{equation}
    In other words, we have found that any $(G,\omega)$-symmetric
    closure tensor~$M$ must obey the condition
    $\mathcal{T}\mkern2mu\ket{M}=\ket{M}$
    where the projection~$\mathcal{T}$ is given by
    \begin{equation}
        \mathcal{T}
        \coloneqq
         \smashoperator[l]{\sum_{g,
                                 h\in
                                 G}}
         \delta_{h
                 g,
                 g
                 h}
         \sum_{\alpha\in
               G}
         \ketbra{\alpha,
                 \alpha
                 g,
                 \alpha
                 g
                 h,
                 \alpha
                 h}
                {\alpha,
                 \alpha
                 g,
                 \alpha
                 g
                 h,
                 \alpha
                 h}.
    \end{equation}
    Thus $\ket{M}$ lives in the common eigenspace%
    \footnote{Note that the projections~$\mathcal{P}^\omega$
        and~$\mathcal{T}$ do \emph{not} necessarily
        commute in the entire vector space of closures, however, they
        \emph{do} commute in their common eigenspace (trivially).}
    of the two
    projections~$\mathcal{P}^\omega$ and~$\mathcal{T}$:
    \begin{equation}
        \mathcal{T}
        \mathcal{P}^\omega
        \ket{M}
        =\ket{M}
        =\mathcal{P}^\omega
         \mathcal{T}\mkern2mu
         \ket{M}.
    \end{equation}
    It is not difficult to verify that the vectors
    \begin{equation}
        \ket{M(g,h)}
        =\sum_{\alpha\in
               G}
         \lambda_M(\alpha;
                   g,
                   h)\mkern2mu
         \ket{\alpha,
              \alpha
              g,
              \alpha
              g
              h,
              \alpha
              h}
    \end{equation}
    span
    the image of the projection~$\mathcal{T}$ \emph{independently} of
    their coefficients~$\lambda_M(\alpha;g,h)$.
    Indeed, one has $\mathcal{T}\ket{M(g,h)}=\ket{M(g,h)}$ and
    $\braket{M(g,h)}{M(g'\mkern-2mu,h')}=\delta_{g,g'}\mkern2mu\delta_{h,h'}\mkern2mu\abs{G}$.
    Since the projection~$\mathcal{P}^\omega$
    will be applied by the triangle tensors automatically, we are done.
\end{proof}

\begin{definition}[$c$-regularity]
    Let $c$ be a twisted 2-cocycle on a finite group~$G$.
    \begin{enumerate}
        \item A pair~$(g,h)\in G\times G$ is called \emph{$c$-regular} if $h\in Z(g)$
            and $c_g(h,x)=c_g(x,h)$ for all $x\in Z(g,h)$.
        \item A pair conjugacy class $\mathcal{C}(g,h)=\{(g^t,h^t)\mid t\in G\}$
            is called \emph{$c$-regular}
            if $(g,h)$ is a $c$-regular pair.
    \end{enumerate}
\end{definition}

\begin{remark}
    This definition of a $c$-regular pair conjugacy class is well-defined:
    if the representative~$(g,h)$ is a $c$-regular pair, then so are all
    other pairs~$(g^t,h^t)$ [Lemma~\ref{lem:c_pair_conjugation}].
    Furthermore, if the twisted 2-cocycle~$c'$ differs from~$c$ by a twisted 2-coboundary
    \begin{equation}
        (\mathrm{d}
         \epsilon_g)(x,
                     y)
        \coloneqq
         \frac{\epsilon_g(x)\mkern2mu
               \epsilon_{x^{-1}
                         g
                         x}(y)}
              {\epsilon_g(x
                          y)}
    \end{equation}
    then $(g,h)$ is a $c$-regular pair if and only if it is $c'$-regular
    [Lemma~\ref{lem:twisted_cohomology}]. Hence the notion of a $c$-regular pair
    conjugacy class only depends on the (twisted) cohomology class of~$c$.
    Obviously, $hg=gh$ holds for any $c$-regular
    pair~$(g,h)$. Thus if the (twisted) cohomology class of~$c$ is trivial 
    then the above definition reduces
    to the commuting pair conjugacy classes of~\cite{Schuch:2010p2806}.\qed
\end{remark}

Let us now focus on the particular twisted 2-cocycle~$c^\omega$ which we
obtain from the original 3-cocycle~$\omega$ by setting
\begin{equation}
    \label{eq:c_omega_definition}
    c_g^\omega(x,
                   y)
    \coloneqq
     \frac{\omega(g,
                  x,
                  y)\mkern2mu
           \omega(x,
                  y,
                  \prescript{xy\mkern-2mu}{}{g})}
          {\omega(x,
                  \prescript{x\mkern-2mu}{}{g},
                  y)}.
\end{equation}
One can easily check that $\omega'\sim\omega$ as 3-cocycles
implies $c^{\omega'}\sim c^\omega$ as twisted 2-cocycles
[Lemma~\ref{lem:c_special}], however, the converse is \emph{not} true%
\footnote{For example, all nontrivial 3-cocycles of~$\mathbb{Z}_n$,
    $\mathbb{Z}_n\times\mathbb{Z}_n$ and $D_m$ for $m$~odd yield
    trivial 2-cocycles~$c^\omega\sim1$.}.

Also one has that $(g,h)$ is a
$c^\omega$-regular pair if and only if $(h,g)$ is, hence the definition
is symmetric [Lemma~\ref{lem:c_omega_regularity}].

\begin{proposition}
    \label{prop:closure_minimal_torus}
    A basis of the closure space of $(G,\omega)$-injective
    tensors on a \emph{minimal} torus corresponds bijectively
    to $c^\omega$-regular pair conjugacy classes~$\mathcal{C}(g,h)$.
\end{proposition}
\begin{proof}
    We will verify the claim
    in three steps. First, we choose certain coefficients~$\lambda_M(\alpha;g,h)$
    such that the virtual $(G,\omega)$-action on the closure states~$\ket{M(g,h)}$
    becomes particularly simple.
    As noted above this does not affect the space spanned by the~$\ket{M(g,h)}$ at all
    (but merely yields different spanning sets for the image of~$\mathcal{T}$).
    Second, we will use these particular states to build manifestly $(G,\omega)$-symmetric
    closure states~$\ket{M'(g,h)}$, i.e. implement the second
    projection~$\mathcal{P}^\omega$ explicitly.
    Third, we will show that any two such states from pairs~$(g,h)$ and~$(g^t,h^t)$
    are identical, that they vanish unless $(g,h)$ is a $c^\omega$-regular pair
    and that states from distinct $c^\omega$-regular pair conjugacy classes are
    linearly independent.
    
    For the first step it will be convenient
    to set
    \begin{equation}
        \label{eq:closure_DW_components}
        \lambda_M(\alpha;
                  g,
                  h)
        \coloneqq
         \frac{\omega(\alpha,
                      g,
                      h)\mkern2mu
               \omega\bigl((g^{-1})^\alpha,
                           \alpha
                           g,
                           h\bigr)}
              {\omega(\alpha,
                      h,
                      g)\mkern2mu
               \omega\bigl((h^{-1})^\alpha,
                           \alpha
                           h,
                           g\bigr)}
    \end{equation}
    for all $hg=gh$.
    Indeed, for this choice we can show that
    \begin{equation}
        \label{eq:closure_action}
        V^\omega(s)\mkern2mu
        \ket{M(g,
               h)}
        =\eta_g(h,
                s)\mkern2mu
         \ket{M(g^s\mkern-2mu,
                h^s)}
    \end{equation}
    holds for any $s\in G$.
    Here we defined for all $x\in Z(g)$ and $y\in G$
    \begin{equation}
        \label{eq:eta_definition}
        \eta_g(x,
               y)
        \coloneqq
         \frac{c_g^\omega(y^{-1}\mkern-3mu,
                              x^y)}
              {c_g^\omega(x,
                              y^{-1})}
        =\frac{\omega(g,
                      y^{-1}\mkern-3mu,
                      x^y)\mkern2mu
               \omega(y^{-1}\mkern-3mu,
                      x^y\mkern-2mu,
                      g^y)\mkern2mu
               \omega(x,
                      g,
                      y^{-1})}
              {\omega(g,
                      x,
                      y^{-1})\mkern2mu
               \omega(y^{-1}\mkern-3mu,
                      g^y,
                      x^y)\mkern2mu
               \omega(x,
                      y^{-1}\mkern-3mu,
                      g^y)}.
    \end{equation}
    The validity of~\eqref{eq:closure_action} can be seen from
    the non-vanishing matrix elements
    \begin{equation*}
        \expval{\alpha,
                \alpha
                g^s\mkern-2mu,
                \alpha
                (g
                 h)^s\mkern-2mu,
                \alpha
                h^s}
               {V^\omega(s)}
               {M(g,
                  h)}
        =\mu(\alpha;
             g,
             h,
             s)\mkern2mu
         \lambda_M(\alpha
                   s;
                   g,
                   h)
    \end{equation*}
    where
    \begin{equation*}
        \mu(\alpha;
            g,
            h,
            s)
        =\frac{\omega\bigl((g^{-1})^{\alpha
                                     s}\mkern-2mu,
                           \alpha
                           g^s\mkern-2mu,
                           s\bigr)\mkern2mu
               \omega\bigl((h^{-1})^{\alpha
                                     s}\mkern-2mu,
                           \alpha
                           (g
                            h)^s\mkern-2mu,
                           s\bigr)}
              {\omega\bigl((h^{-1})^{\alpha
                                     s}\mkern-2mu,
                           \alpha
                           h^s\mkern-2mu,
                           s\bigr)\mkern2mu
               \omega\bigl((g^{-1})^{\alpha
                                     s}\mkern-2mu,
                           \alpha
                           (g
                            h)^s\mkern-2mu,
                           s\bigr)},
    \end{equation*}
    and the fact that
    \begin{equation*}
        \frac{\mu(\alpha;
                  g,
                  h,
                  s)\mkern2mu
              \lambda_M(\alpha
                        s;
                        g,
                        h)}
             {\lambda_M(\alpha;
                        g^s\mkern-2mu,
                        h^s)}
        =\eta_g(h,
                s)
    \end{equation*}
    [Lemma~\ref{lem:eta_neutral}].
    
    As a second step, we may simply use~\eqref{eq:closure_action} in order to
    apply the second projection~$\mathcal{P}^\omega$ and obtain the manifestly $(G,\omega)$-symmetric
    states
    \begin{equation}
        \ket{M'(g,
                h)}
        \coloneqq
         \frac{1}
              {\abs{G}}
         \sum_{s\in
               G}
         \eta_g(h,
                s)\mkern2mu
         \ket{M(g^s\mkern-2mu,
                h^s)}.
    \end{equation}
    
    Turning to the third step, we find that
    \begin{equation*}
        \ket{M'(g^t\mkern-2mu,
                h^t)}
        =\frac{1}
              {\eta_g(h,
                      t)}\mkern2mu
         \ket{M'(g,
                 h)}
    \end{equation*}
    for all $t\in G$ [Lemma~\ref{lem:eta_property1}],
    hence all states within the same commuting pair conjugacy
    class are identical. As a trivial consequence, should $\ket{M'(g,h)}$
    vanish for a certain commuting pair~$(g,h)$ then it does so for the entire
    pair conjugacy class~$\mathcal{C}(g,h)$.
    But are there states~$\ket{M'(g,h)}$ which actually vanish?
    We obtain
    \begin{equation}
        \label{eq:overlap}
        \braket{M'(g,
                   h)}
               {M'(j,
                   k)}
        =\sum_{s\in
               G}
         \frac{1}
              {\eta_g(h,
                      s)}\mkern2mu
         \delta_{j,
                 g^s}\mkern2mu
         \delta_{k,
                 h^s}
    \end{equation}
    for all $hg=gh$ and $kj=jk$ [Lemma~\ref{lem:eta_property1}]. In particular, we have the norm
    \begin{equation*}
        \norm[\big]{\ket{M'(g,
                            h)}}^2
        =\smashoperator[l]{\sum_{s\in
                                 Z(g,
                                   h)}}
         \frac{1}
              {\eta_g(h,
                      s)} \\
        =\begin{cases}
             \abs{Z(g,h)} & \text{if $(g,h)$ is $c^\omega$-regular}, \\
                  0       & \text{else},
         \end{cases}
    \end{equation*}
    where the last equality follows from~\cite{Dijkgraaf:1990p2241,Hu:mIFulyz6}.
    In other words, the state~$\ket{M'(g,h)}$ vanishes
    unless the pair~$(g,h)$ is $c^\omega$-regular.
    Hence focussing
    on $c^\omega$-regular pairs~$(g,h)$ and~$(j,k)$ it is obvious
    from~\eqref{eq:overlap} that closure states from distinct
    $c^\omega$-regular pair conjugacy classes are orthogonal, thus
    linearly independent.
\end{proof}

\begin{theorem}[Closure space]
    \label{thm:closure_torus}
    A basis of the closure space of $(G,\omega)$-injective
    tensors on an \emph{arbitrary} torus corresponds bijectively
    to $c^\omega$-regular pair conjugacy classes~$\mathcal{C}(g,h)$.
    
    Furthermore, the dimension of this closure space only depends on the twisted cohomology
    class of~$c^\omega$.
\end{theorem}
\begin{proof}
    We first show the analogue of Lemma~\ref{lem:closure_minimal_torus_span}.
    Define the \enquote{cross} tensor~$m(g,h)$ by
    \begin{equation}
        \lambda_m(\alpha,
                  \beta,
                  \gamma;
                  g,
                  h)
        \coloneqq
         \vcorrect{48bp}{\includegraphics[scale=0.25]{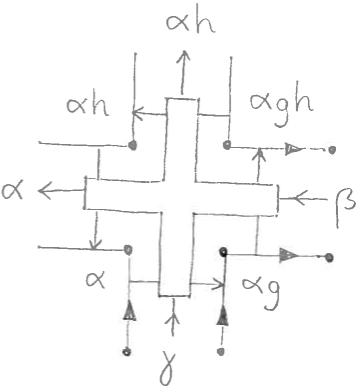}}
        \coloneqq
         \frac{\omega(\alpha,
                      g,
                      h)}
              {\omega(\alpha,
                      h,
                      g)}\mkern2mu
         \frac{\omega(\beta
                      g^{-1}
                      \alpha^{-1}\mkern-2mu,
                      \alpha
                      g,
                      h)}
              {\omega(\alpha
                      \gamma^{-1}\mkern-2mu,
                      \gamma,
                      g)},
    \end{equation}
    with all other elements being zero. Note that this tensor depends
    on the orientation of the virtual boundary at the input side of
    each \textsc{MPO} ring. It is possible to define this tensor
    consistently for other orientations than the one shown here.
    
    \begin{figure}
        \centering
        \includegraphics[scale=0.25]{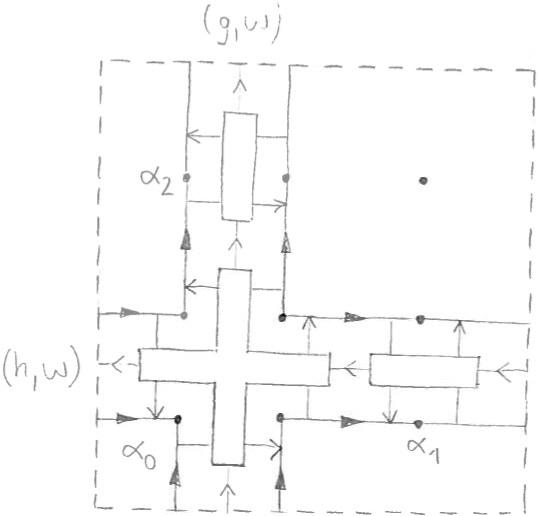}\hfil
        \includegraphics[scale=0.25]{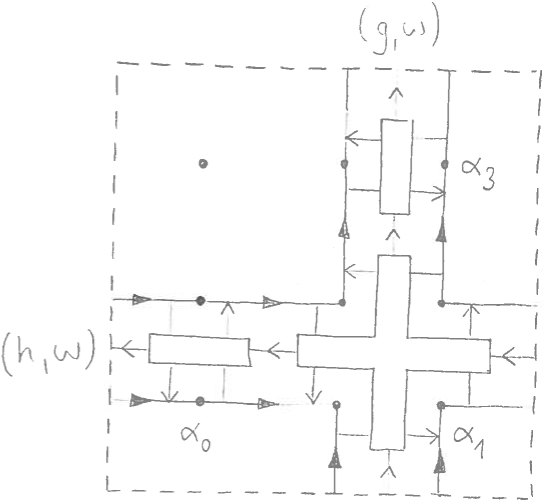}
        \caption
            [Transparency]
            {The closure tensors~$M_0(g,h)$ and~$M_1(g,h)$.}
        \label{fig:closure_tensors}
    \end{figure}
    
    We can evaluate the elements of the closure tensor~$M_0(g,h)$
    in Figure~\ref{fig:closure_tensors}
    by contracting its building blocks. Let us call the result
    $\lambda_{M_0}(\alpha_0,\alpha_1,\alpha_2;g,h)$.
    Using the twisted virtual symmetry one may check that~$M_0(g,h)$ can be deformed freely to other
    closure tensors like~$M_1(g,h)$ \emph{without} any change
    to its building blocks.
    Similarly, one can define $P$- and $Q$-like closures and
    show that any closure can be deformed freely into any other,
    for any fixed cellulation of the torus with $(G,\omega)$-injective
    tensors. This proves the
    $\supset$~part. We leave the $\subset$~part to the reader.
    
    \begin{figure}
        \centering
        \includegraphics[scale=0.25]{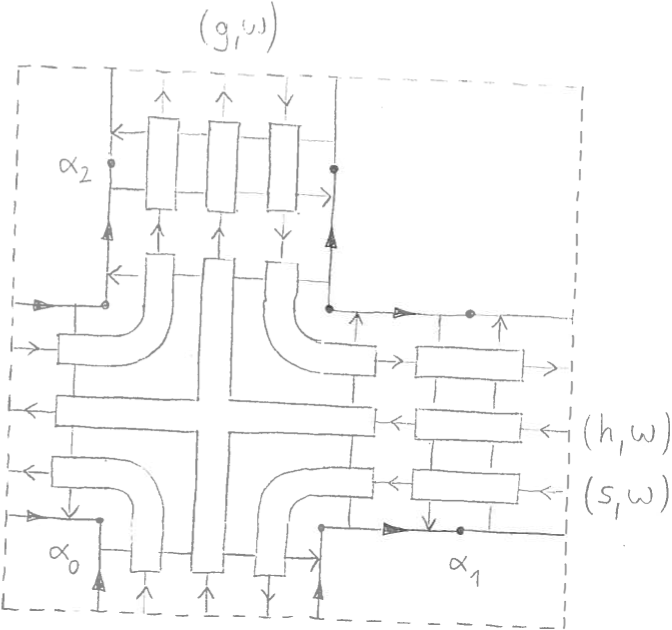}
        \caption
            {The $(G,\omega)$-action on the closure state~$\ket{M_0(g,h)}$.}
        \label{fig:closure_action}
    \end{figure}
    
    We now extend Proposition~\ref{prop:closure_minimal_torus}, i.e.
    we show that the $(G,\omega)$-action on the closure states
    of an arbitrary torus is identical to~\eqref{eq:closure_action},
    which immediately yields the claimed basis of the closure space.
    Without loss of generality, consider the closure state
    \begin{equation*}
        \ket{M_0(g,
                 h)}
        =\smashoperator[l]{\sum_{\alpha_i\in
                                 G}}
         \lambda_{M_0}(\alpha_0,
                       \alpha_1,
                       \alpha_2;
                       g,
                       h)\mkern2mu
         \ket{\alpha_0,
              \alpha_1,
              \alpha_0
              g,
              \alpha_2
              g,
              \alpha_0
              g
              h,
              \alpha_1
              h,
              \alpha_0
              h,
              \alpha_2},
    \end{equation*}
    see Figure~\ref{fig:closure_tensors}. Applying the twisted virtual symmetry as
    in Figure~\ref{fig:closure_action}
    one may show that
    \begin{equation*}
        \frac{\expval{\alpha_0,
                      \alpha_1,
                      \alpha_0
                      g^s\mkern-2mu,
                      \alpha_2
                      g^s\mkern-2mu,
                      \alpha_0
                      (g
                       h)^s\mkern-2mu,
                      \alpha_1
                      h^s\mkern-2mu,
                      \alpha_0
                      h^s\mkern-2mu,
                      \alpha_2}
                     {V^\omega(s)}
                     {M_0(g,
                          h)}}
             {\lambda_{M_0}(\alpha_0,
                            \alpha_1,
                            \alpha_2;
                            g^s\mkern-2mu,
                            h^s)}
        =\eta_g(h,
                s),
    \end{equation*}
    which indeed implies~\eqref{eq:closure_action}.
    We can then construct the symmetric closure states~$\{\ket{M_0'(g,h)}\}$
    and select a basis like in Proposition~\ref{prop:closure_minimal_torus}.
\end{proof}

\begin{remark}
    Observe that we can obtain a closure of the minimal torus
    from the \enquote{cross} tensor~$m(g,h)$ via
    \begin{equation*}
        \vcorrect{44bp}{\includegraphics[scale=0.25]{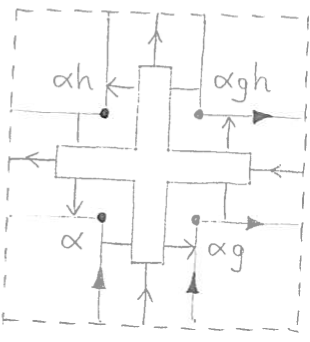}}
        =\lambda_m(\alpha,
                   \alpha,
                   \alpha
                   h;
                   g,
                   h)
        =\lambda_M(\alpha;
                   g,
                   h).
    \end{equation*}
    This illuminates our particular choice of coefficients in the proof
    of Proposition~\ref{prop:closure_minimal_torus}.\qed
\end{remark}

\begin{example}
    Let $G=\mathbb{Z}_2\times\mathbb{Z}_2\times\mathbb{Z}_2$
    and $\mathcal{L}^\omega$ the closure space of $(G,\omega)$-injective
    tensors on a torus. There exists a 3-cocycle~$\omega$ such that
    \begin{equation*}
        \dim
        \mathcal{L}^\omega
        =22
        <\abs{G}^2
        =\dim
         \mathcal{L}.
    \end{equation*}
    This clearly shows that the closure space of
    $(G,\omega)$-injective tensors may be significantly smaller
    than the one for $G$-injective tensors, i.e. the untwisted
    case.\qed
\end{example}

\section{Parent Hamiltonians}
\label{sec:hamiltonians}

Any \textsc{PEPS}~$\ket{\psi}$ whose local tensor for region~$R$
is~$(A_i)$ has a reduced density operator~$\rho_R$ with
$\mathrm{supp}(\rho_R)\subset\mathcal{L}_R$.
It is then clear that every Hamiltonian
\begin{equation}
    \label{eq:parent_hamiltonian}
    H
    =\sum_v
     h_v
\end{equation}
where $h_v\geq0$ acts on the region~$R_v$, i.e. all polygons around the vertex~$v$
of the cellulation, is a parent Hamiltonian of the \textsc{PEPS}~$\ket{\psi}$
if $\mathcal{L}_v\subseteq\ker(h_v)$. Furthermore, such a Hamiltonian
is frustration-free, i.e. $h_v\mkern2mu\ket{\psi}=0$.

In order to grow local ground states into global ground states while keeping
their properties under tight control we need the following result. 

\begin{theorem}[Intersection]
    \label{thm:intersection}
    Let $R_1$ and $R_2$ be regions with a nontrivial
    intersection~$R=R_1\cap R_2$. Let $(A_i)$, $(B_j)$, $(C_k)$ be compatible
    $(G,\omega)$-injective tensors such that the contraction~$(A_iB_j)$
    corresponds to~$R_1$, $(B_jC_k)$ to~$R_2$ and~$B_j$ to~$R$.
    
    Then
    \begin{equation}
        (\mathcal{L}_{R_1}\otimes
         \mathcal{H}_{R_2\setminus
                      R})\cap
        (\mathcal{H}_{R_1\setminus
                      R}\otimes
         \mathcal{L}_{R_2})
        =\mathcal{L}_{R_1\cup
                      R_2}.
    \end{equation}
\end{theorem}
\begin{proof}
    \begin{figure}
        \centering
        \includegraphics[scale=0.25]{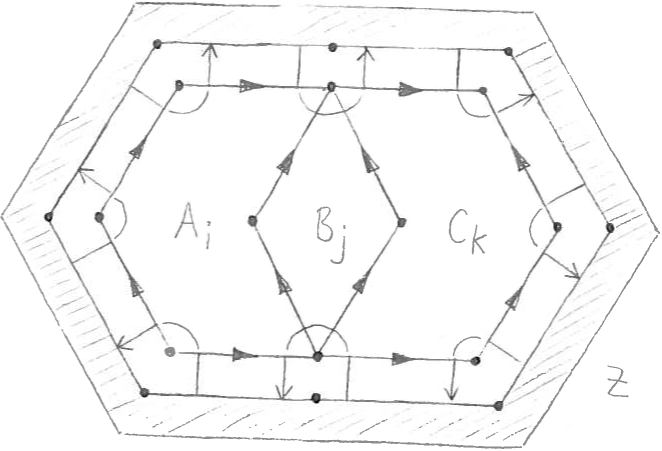}
        \caption
            [Intersection]
            {Intersecting regions~$R_1$ and $R_2$ corresponding to
             $(G,\omega)$-injective tensors~$(A_iB_j)$ and~$(B_jC_k)$
             respectively. Each such region~$R_v$ is naturally associated
             with a vertex~$v$ of the cellulation.}
        \label{fig:intersection}
    \end{figure}

    Without loss of generality let us focus on regions~$R_v$
    as shown in Figure~\ref{fig:intersection}. We will only prove the
    $\subset$~part of the claim since the other inclusion is obvious.
    
    Let $\ttr(A_i,B_j,C_k,Z)$ denote
    the element of the contracted tensor of Figure~\ref{fig:intersection}.
    Then we clearly have
    \begin{align}
        \mathcal{L}_{R_1\cup
                     R_2}
        & =\Bigl\{\sum_{i,
                        j,
                        k}
                  \ttr(A_i,
                       B_j,
                       C_k,
                       Z)\mkern2mu
                  \ket{i,
                       j,
                       k}\Bigm\vert
                  Z\Bigr\}, \\
        \intertext{and}
        \mathcal{L}_{R_1}\otimes
        \mathcal{H}_{R_2\setminus
                     R_1}
        & =\Biggl\{\sum_{i,
                         j,
                         k}
                   \vcorrect{56bp}{\includegraphics[scale=0.25]{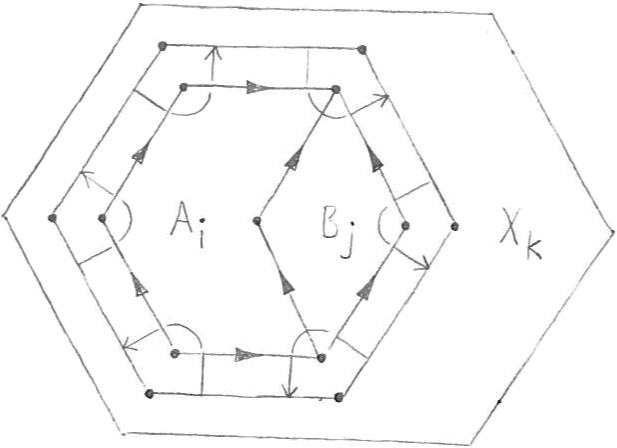}}\mkern3mu
                   \ket{i,
                        j,
                        k}\Biggm\vert
                   X_k\Biggr\},
           \label{eq:subspace1} \\
        \mathcal{H}_{R_1\setminus
                     R_2}\otimes
        \mathcal{L}_{R_2}
        & =\Biggl\{\sum_{i,
                         j,
                         k}
                   \vcorrect{56bp}{\includegraphics[scale=0.25]{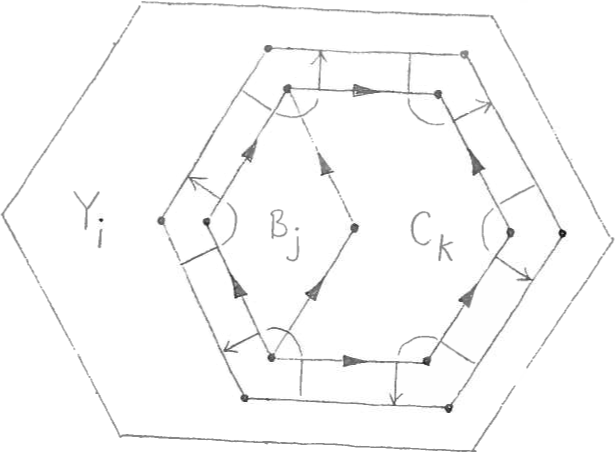}}\mkern3mu
                   \ket{i,
                        j,
                        k}\Biggm\vert
                   Y_i\Biggr\}.
           \label{eq:subspace2}
    \end{align}
    As before we may assume that
    $(X_k)$, $(Y_i)$ and~$Z$ are $(G,\omega)$-symmetric. In order to
    simplify notation set
    \begin{align*}
        \braket{\alpha_1,
                \dots,
                \alpha_6}
               {X_k}
        & \coloneqq
           \vcorrect{56bp}{\includegraphics[scale=0.25]{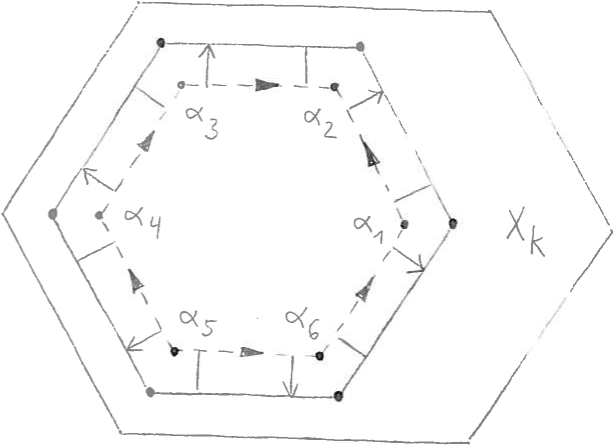}}\mkern3mu.
    \end{align*}
    Intersecting the subspaces~\eqref{eq:subspace1} and~\eqref{eq:subspace2}
    we obtain
    \begin{align*}
        & \braket{\alpha_1,
                  \dots,
                  \alpha_6}
                 {X_k} \\
        & =\sum_{i,
                 j}\mkern3mu
           \vcorrect{48bp}{\includegraphics[scale=0.25]{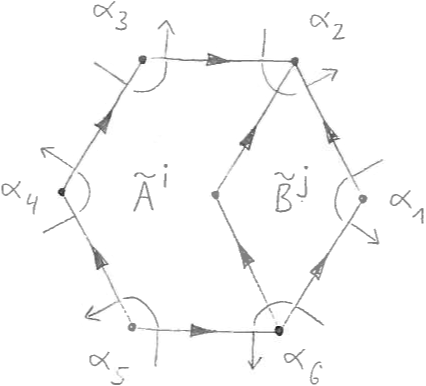}}\mkern9mu
           \vcorrect{56bp}{\includegraphics[scale=0.25]{figure44}}
           \displaybreak[0] \\
        & =\sum_{i,
                 j}\mkern3mu
           \vcorrect{49bp}{\includegraphics[scale=0.25]{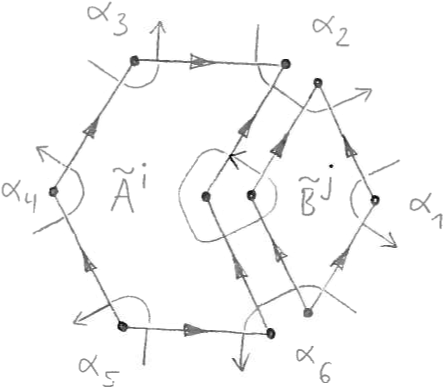}}\mkern9mu
           \vcorrect{56bp}{\includegraphics[scale=0.25]{figure45}}
           \displaybreak[0] \\
        & =\sum_{i,
                 \beta}\mkern3mu
           \vcorrect{49bp}{\includegraphics[scale=0.25]{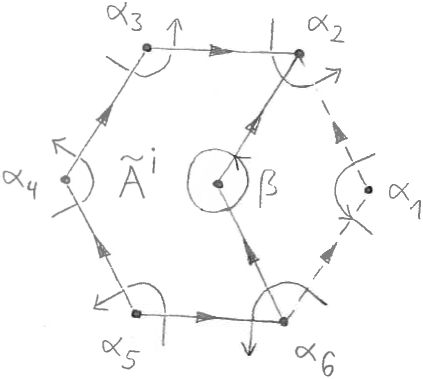}}\mkern9mu
           \vcorrect{70bp}{\includegraphics[scale=0.25]{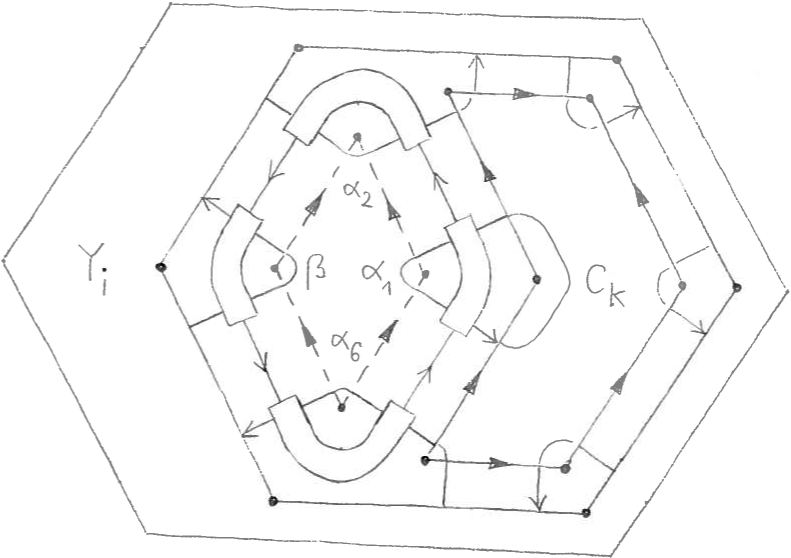}}
           \displaybreak[0] \\
        & =\sum_{i,
                 \beta}\mkern3mu
           \vcorrect{49bp}{\includegraphics[scale=0.25]{figure49}}\mkern9mu
           \vcorrect{70bp}{\includegraphics[scale=0.25]{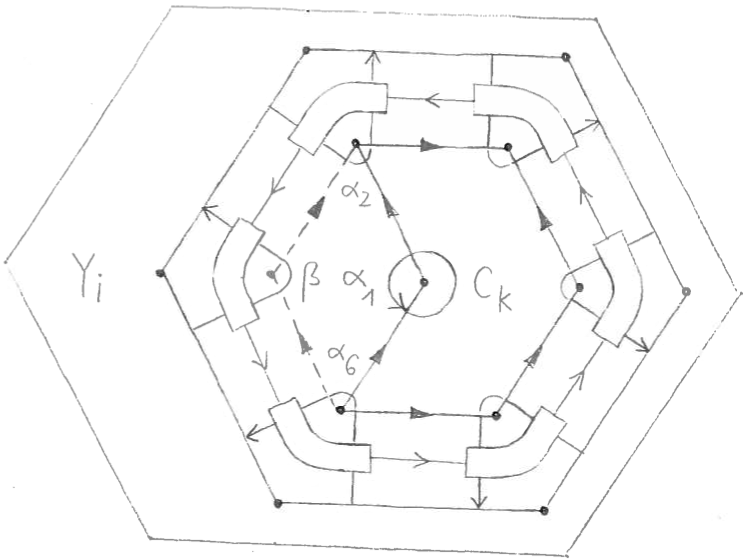}}
           \displaybreak[0] \\
        & =\sum_{i,
                 \beta}\mkern3mu
           \vcorrect{49bp}{\includegraphics[scale=0.25]{figure49}}\mkern9mu
           \vcorrect{56bp}{\includegraphics[scale=0.25]{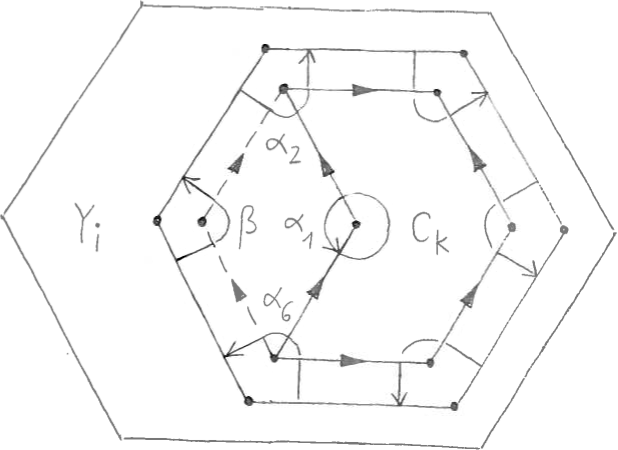}} \\
        & =\braket{\alpha_1,
                   \dots,
                   \alpha_6}
                  {(C_k
                    Z)}
    \end{align*}
    where we defined~$Z$ via
    \begin{multline*}
        \braket{\alpha_1,
                \dots,
                \alpha_8}
               {Z}
        \coloneqq
         \sum_{i,
               \beta}\mkern3mu
         \vcorrect{52bp}{\includegraphics[scale=0.25]{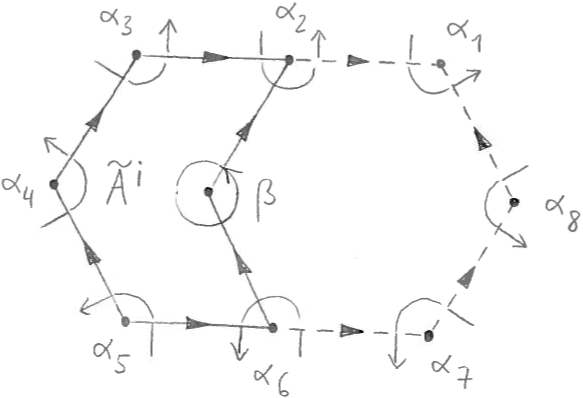}} \\
         \vcorrect{56bp}{\includegraphics[scale=0.25]{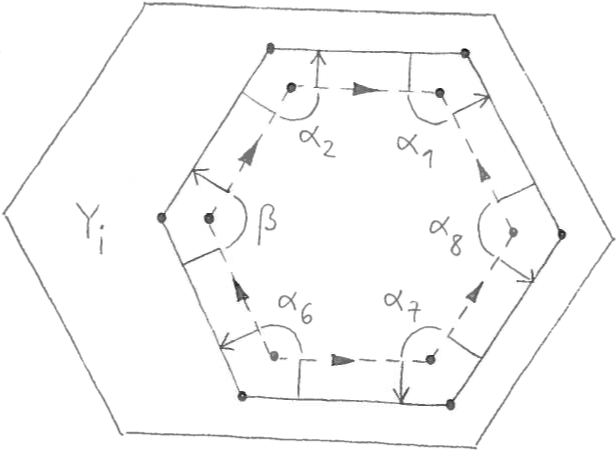}}.
    \end{multline*}
\end{proof}

Now if the local kernels of the Hamiltonian \emph{coincide} with
the local subspaces generated by the tensors, i.e. $\ker(h_v)=\mathcal{L}_v$, then
Theorem~\ref{thm:intersection} ensures that no undesired states
occur in the (local) ground state space as we grow a region~$R_v$
step by step until we cover the whole surface. Hence on a torus
the ground states of such a parent Hamiltonian are exactly
given by the closure space of Theorem~\ref{thm:closure_torus}.

\section{Twisted Isometry}
\label{sec:twisted_isometry}

\subsection{Standard Form}

As has been observed in~\cite{Schuch:2011p3031,Schwarz:2012vn},
any \textsc{PEPS} tensor can be decomposed into a partial isometry~$W$
and a deformation~$Q$. In the context of twisted virtual symmetry this means

\begin{lemma}
    Every $(G,\omega)$-injective tensor~$A$ can be written as
    \begin{equation}
        \label{eq:isometric_form}
        A
        =Q
         \mathcal{P}^\omega
    \end{equation}
    where $Q$ is a positive map.
\end{lemma}
\begin{proof}
    We may restrict the map
    $A\colon\mathbb{C}G\otimes\dots\otimes\mathbb{C}G\to\mathcal{H}$ to
    $\mathcal{A}\colon\image(\mathcal{P}^\omega)\to\mathcal{H}$ so that
    $A=\mathcal{A}\mathcal{P}^\omega$. Then $\mathcal{A}$ is clearly invertible.
    The polar decomposition~$\mathcal{A}=UQ$ with
    $Q=\sqrt{\mathcal{A}^\dagger\mathcal{A}}$ positive and
    $U=\mathcal{A}Q^{-1}$ unitary then yields $A\simeq Q\mathcal{P}^\omega$
    because we are free to ignore local unitaries at the physical level.
\end{proof}

\begin{definition}[Standard form]
    The projection~$\mathcal{P}^\omega$ is called the $(G,\omega)$-isometric
    standard form.
\end{definition}

Let us focus on $(G,\omega)$-isometric tensors~$A$ for the moment.

\subsection{Dijkgraaf-Witten Form}

The concept of $(G,\omega)$-\emph{isometric} \emph{triangle} tensors
has a very nice interpretation in terms of time slices in (discrete)
\textsc{Dijkgraaf-Witten} topological quantum field theory in $(2+1)$
dimensions~\cite{Dijkgraaf:1990p2241}.
Its central ingredient is the partition function which is constructed
from a branched triangulation of the space-time 3-manifold by colouring
the edges and assigning weights to the tetrahedra of the triangulation.
Let us briefly review how these weights are defined.

Let $\Delta_{ijkl}$ be a tetrahedron with ordered vertices $i<j<k<l$. Then
a branching structure is given as follows: for any two vertices $i<j$ draw
an oriented edge~$(i\to j)$ from~$i$ to $j$. Any such branched tetrahedron
can be assigned an orientation~$\sgn(\Delta_{ijkl})=\pm1$ by looking at the face~$\Delta_{jkl}$
formed by the vertices $j<k<l$ from the direction of the smallest vertex~$i$: if
the majority of edges points in counterclockwise direction we set $\sgn(\Delta_{ijkl})\coloneqq1$,
otherwise $\sgn(\Delta_{ijkl})\coloneqq-1$.
We can colour the tetrahedron~$\Delta_{ijkl}$ by assigning
$g_{ji}\in G$ to each oriented edge~$(i\to j)$.
Naturally we set $g_{ij}\coloneqq g_{ji}^{-1}$. Furthermore we require this colouring
to have \emph{flat connections} everywhere, i.e. each branched face~$\Delta_{ijk}$
with vertices $i$, $j$ and $k$ satisfies
\begin{equation}
    g_{ij}
    g_{jk}
    g_{ki}
    =e.
\end{equation}
Finally we assign the following amplitude to any branched, coloured
tetrahedron~$\Delta_{ijkl}$:
\begin{equation}
    \psi(\Delta_{ijkl})
    \coloneqq
     \omega(g_{lk},
            g_{kj},
            g_{ji})^{\sgn(\Delta_{ijkl})}.
\end{equation}

\bigskip
We can define a $(G,\omega)$-isometric triangle tensor via
the amplitude of its associated tetrahedron~$\Delta_{0ijk}$. In the following
we set $g_i\coloneqq g_{i0}$ and focus on a particular vertex
ordering without loss of generality.
\begin{definition}[\textsc{Dijkgraaf-Witten} form]
    Let $g_1,g_2,g_3\in G$. Then
    \begin{equation}
        \vcorrect{29bp}{\includegraphics[scale=0.25]{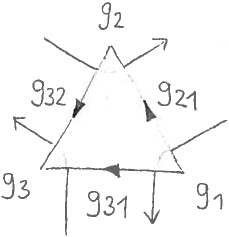}}
        \coloneqq
         \vcorrect{38bp}{\includegraphics[scale=0.25]{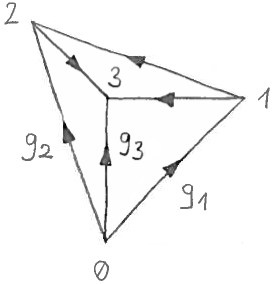}}.
    \end{equation}
    with the physical index $(g_{21},g_{32},g_{31})$ is called a
    \textsc{Dijkgraaf-Witten} triangle tensor.
\end{definition}

In other words, the above tensor defines the map
\begin{equation}
    A_\mathrm{DW}^\omega
    =\smashoperator[l]{\sum_{g_i\in
                             G}}
     \omega(g_3
            g_2^{-1}\mkern-3mu,
            g_2
            g_1^{-1}\mkern-3mu,
            g_1)^{-1}\mkern2mu
     \ketbra{g_2
             g_1^{-1}\mkern-3mu,
             g_3
             g_2^{-1}\mkern-3mu,
             g_3
             g_1^{-1}}
            {g_1,
             g_2,
             g_3}
\end{equation}
due to the flat connection condition $g_{ij}=g_ig_j^{-1}$ for each of the three side faces
of the tetrahedron. It is not difficult to prove that every such
tensor is indeed $(G,\omega)$-isometric.

\begin{lemma}
    Every \textsc{Dijkgraaf-Witten} triangle tensor is
    $(G,\omega)$-isometric.
\end{lemma}
\begin{proof}
    The virtual $(G,\omega)$-symmetry of~$A_\mathrm{DW}^\omega$ is obvious from the Pachner 4--1 move
    \begin{equation*}
        \vcorrect{50bp}{\includegraphics[scale=0.25]{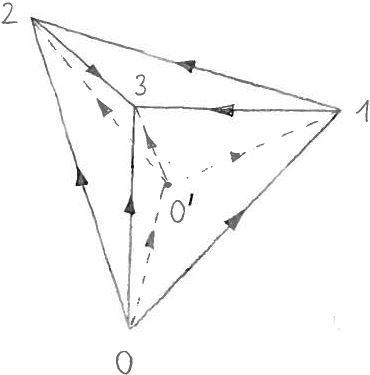}}
        =\vcorrect{38bp}{\includegraphics[scale=0.25]{figure56}}
        =\vcorrect{29bp}{\includegraphics[scale=0.25]{figure55}}.
    \end{equation*}
    Furthermore, we can invert~$A_\mathrm{DW}^\omega$ (on the image of~$\mathcal{P}^\omega$)
    using its adjoint. This is easily seen once we glue
    the tetrahedron representing~$A_\mathrm{DW}^\omega$ and
    its mirror image representing~$(A_\mathrm{DW}^\omega)^\dagger$ along their
    physical faces and apply a Pachner 2--3 move:
    \begin{align*}
        (A_\mathrm{DW}^\omega)^\dagger
        A_\mathrm{DW}^\omega
        & =\sum_{g_i,
                 g_i'\in
                 G}
           \vcorrect{60bp}{\includegraphics[scale=0.25]{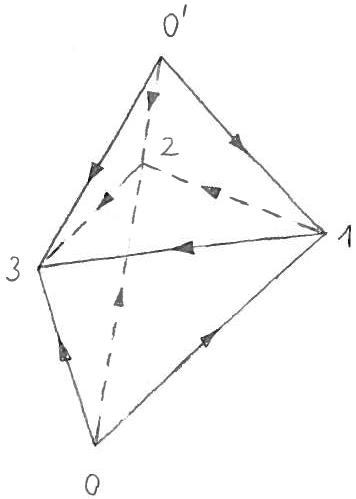}}\mkern6mu
           \ketbra{g_1',
                   g_2',
                   g_3'}
                  {g_1,
                   g_2,
                   g_3}
           \displaybreak[0] \\
        & =\sum_{g\in
                 G}
           \sum_{g_i\in
                 G}
           \vcorrect{60bp}{\includegraphics[scale=0.25]{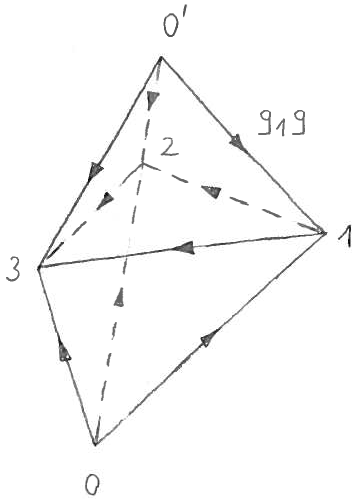}}\mkern6mu
           \ketbra{g_1
                   g,
                   g_2
                   g,
                   g_3
                   g}
                  {g_1,
                   g_2,
                   g_3} \\
        & =\sum_{g\in
                 G}
           \sum_{g_i\in
                 G}
           \vcorrect{60bp}{\includegraphics[scale=0.25]{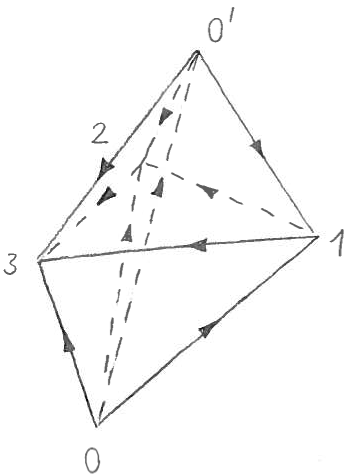}}\mkern6mu
           \ketbra{g_1
                   g,
                   g_2
                   g,
                   g_3
                   g}
                  {g_1,
                   g_2,
                   g_3}
           \displaybreak[0] \\
        & =\sum_{g\in
                 G}
           \sum_{g_i\in
                 G}
           \frac{\omega(g_3
                        g_1^{-1}\mkern-3mu,
                        g_1
                        g,
                        g^{-1})}
                {\omega(g_2
                        g_1^{-1}\mkern-3mu,
                        g_1
                        g,
                        g^{-1})\mkern2mu
                 \omega(g_3
                        g_2^{-1}\mkern-3mu,
                        g_2
                        g,
                        g^{-1})}\mkern2mu
           \ketbra{g_1
                   g,
                   g_2
                   g,
                   g_3
                   g}
                  {g_1,
                   g_2,
                   g_3} \\
        & =\abs{G}\mkern2mu
           \mathcal{P}^\omega
    \end{align*}
    where $g^{-1}\coloneqq g_{0'0}$. This proves the claim.
\end{proof}

We can now easily translate between the standard form~$\mathcal{P}^\omega$
and the form~$A_\mathrm{DW}^\omega$ by applying a unitary map at the
physical level. In the forward direction this unitary map is simply
$A_\mathrm{DW}^\omega$ itself since $A_\mathrm{DW}^\omega\mathcal{P}^\omega=A_\mathrm{DW}^\omega$,
in the backward direction it is~$(A_\mathrm{DW}^\omega)^\dagger/\abs{G}$ because
of the above.

\subsection{Cohomology Classes}

What happens if we change~$\omega\mapsto\omega'=\omega\mkern2mu\mathrm{d}\phi$ in
a $(G,\omega)$-isometric tensor?
Let us first look at a \textsc{Dijkgraaf-Witten} triangle
tensor for~$\omega'$.
As is well known,
the 3-coboundary~$\mathrm{d}\phi$ consists of four 2-cochains~$\phi$ which can naturally
be attached to the faces of the tetrahedron~$\Delta_{0ijk}$. Distinguishing
between the physical face and the virtual ones we obtain
\begin{equation}
    A_\mathrm{DW}^{\omega'}
    =U_\phi
     A_\mathrm{DW}^\omega
     V_\phi
\end{equation}
where
\begin{align}
    U_\phi
    & \coloneqq
       \smashoperator[l]{\sum_{\alpha,
                               \beta,
                               \gamma\in
                               G}}
       \phi(\beta,
            \alpha)\mkern2mu
       \ketbra{\alpha,
               \beta,
               \gamma}
              {\alpha,
               \beta,
               \gamma},
       \displaybreak[0] \\
    V_\phi
    & \coloneqq
       \smashoperator[l]{\sum_{\alpha,
                               \beta,
                               \gamma\in
                               G}}
       \frac{\phi(\gamma
                  \alpha^{-1}\mkern-3mu,
                  \alpha)}
            {\phi(\gamma
                  \beta^{-1}\mkern-3mu,
                  \beta)\mkern2mu
             \phi(\beta
                  \alpha^{-1}\mkern-3mu,
                  \alpha)}\mkern2mu
       \ketbra{\alpha,
               \beta,
               \gamma}
              {\alpha,
               \beta,
               \gamma},
\end{align}
are unitaries whose underlying branching structure is that
of~$A_\mathrm{DW}^\omega$.

It is now straightforward to check that the standard forms
for $(G,\omega)$- and $(G,\omega')$-isometric tensors are related by
\begin{equation}
    \label{eq:projection_gauge}
    \mathcal{P}^{\omega'}
    =V_\phi^\dagger
     \mathcal{P}^\omega
     V_\phi
    \simeq
     \mathcal{P}^\omega
     V_\phi
\end{equation}
since we are allowed to ignore local unitaries at the physical level.
Note that we cannot pull $V_\phi$ to the physical level since it does not
necessarily commute with~$\mathcal{P}^\omega$. However, upon contracting
the tensors with each other and the closure tensors, the $V_\phi$
will cancel at the virtual level.
This establishes the following
\begin{proposition}
    \label{prop:equivalence_isometric}
    Any transformation $\omega\mapsto\omega'\sim\omega$ corresponds
    to local unitary transformations of $(G,\omega)$-isometric \textsc{PEPS}.
\end{proposition}

\subsection{Parent Hamiltonians}

\begin{proposition}
    Parent Hamiltonians of $(G,\omega)$-isometric \textsc{PEPS}
    consist of commuting projections.
\end{proposition}
\begin{proof}
    Let $A_v$ be the $(G,\omega)$-isometric map associated with the
    region~$R_v$. Note that we may need to contract several
    $(G,\omega)$-isometric \textsc{PEPS} tensors to obtain~$A_v$.
    It is not difficult to show that the interaction terms~$h_v$
    \begin{equation*}
        h_v
        \coloneqq
         \id_v-
         A_v
         A_v^\dagger
    \end{equation*}
    are Hermitian projections with $\ker{h_v}=\mathcal{L}_v$.
    
    Now assume that two overlapping interaction terms~$h_1$ and~$h_2$
    are associated with the regions~$R_1$ and $R_2$ of Figure~\ref{fig:intersection}.
    Let us abbreviate $\tilde{h}_v\coloneqq A_vA_v^\dagger$.
    We have
    \begin{equation*}
        \expval{i'\mkern-3mu,
                j'\mkern-3mu,
                k'}
               {\tilde{h}_1}
               {i,
                j,
                k}
        =\delta_{k
                 k'}
         \sum_{\alpha_l,
               \gamma}\mkern3mu
         \vcorrect{47bp}{\includegraphics[scale=0.25]{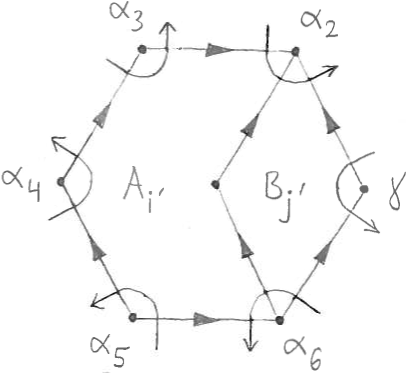}}\mkern12mu
         \vcorrect{47bp}{\includegraphics[scale=0.25]{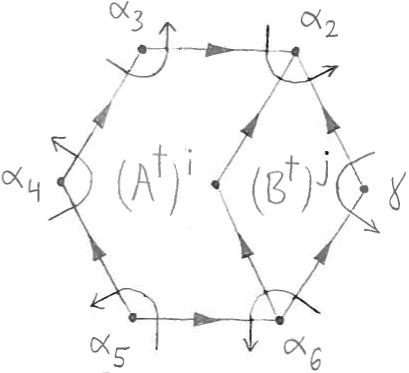}}
    \end{equation*}
    and $\expval{i'\mkern-3mu,j'\mkern-3mu,k'}{\tilde{h}_2}{i,j,k}$
    similarly. Then
    \begin{align*}
        & \expval{i'\mkern-3mu,
                  j'\mkern-3mu,
                  k'}
                 {\tilde{h}_1
                  \tilde{h}_2}
                 {i,
                  j,
                  k} \\
        & =\smashoperator[l]{\sum_{\tilde{\alpha}_l,
                                   \tilde{\beta},
                                   \tilde{\gamma}}}
           \sum_{\alpha_l,
                 \beta,
                 \gamma}
           \vcorrect{47bp}{\includegraphics[scale=0.25]{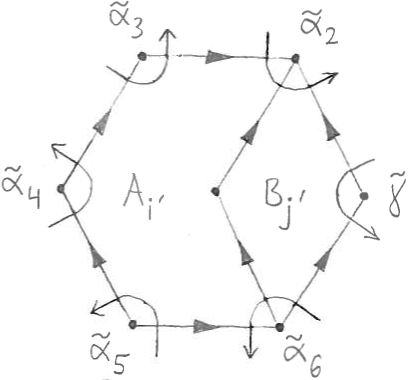}}\mkern9mu
           \vcorrect{49bp}{\includegraphics[scale=0.25]{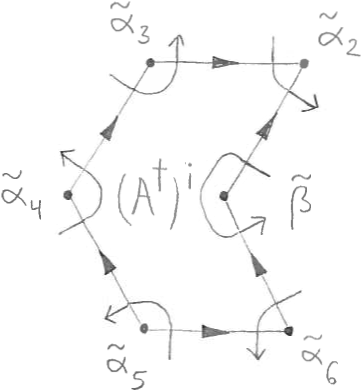}} \\
        & \hphantom{={}}
           \Biggl(\sum_l\mkern3mu
                  \vcorrect{54bp}{\includegraphics[scale=0.25]{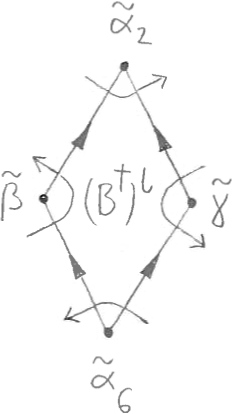}}\mkern12mu
                  \vcorrect{54bp}{\includegraphics[scale=0.25]{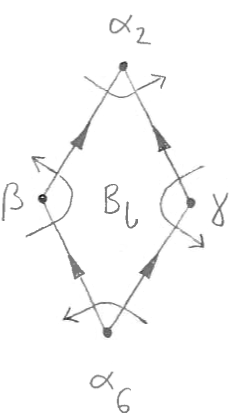}}\Biggr)
           \vcorrect{49bp}{\includegraphics[scale=0.25]{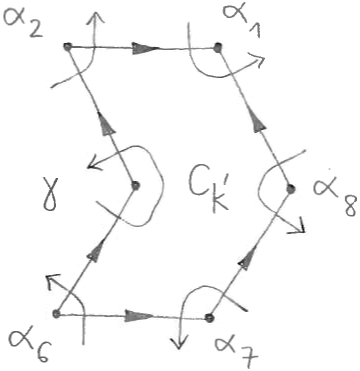}}\mkern9mu
           \vcorrect{48bp}{\includegraphics[scale=0.25]{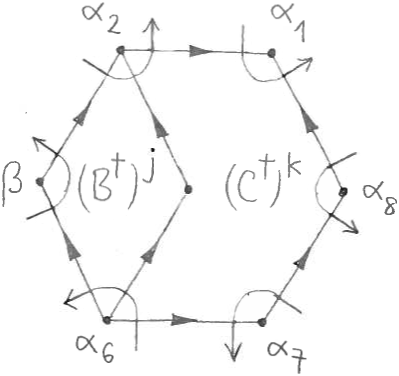}}
           \displaybreak[0] \\
        & =\frac{1}
                {\abs{G}}
           \sum_{\tilde{\alpha}_3,
                 \tilde{\alpha}_4,
                 \tilde{\alpha}_5}
           \sum_{\alpha_l,
                 \beta,
                 \gamma}
           \sum_{g\in
                 G}
           \vcorrect{47bp}{\includegraphics[scale=0.25]{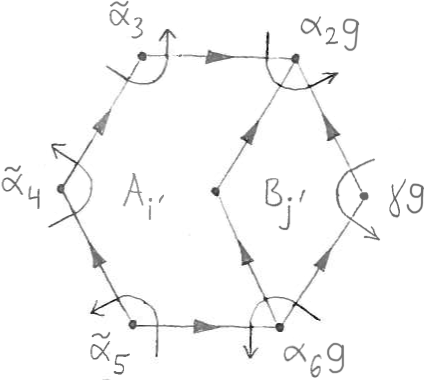}}
           \vcorrect{49bp}{\includegraphics[scale=0.25]{figure67}} \\
        & \hphantom{={}}
           \vcorrect{62bp}{\includegraphics[scale=0.25]{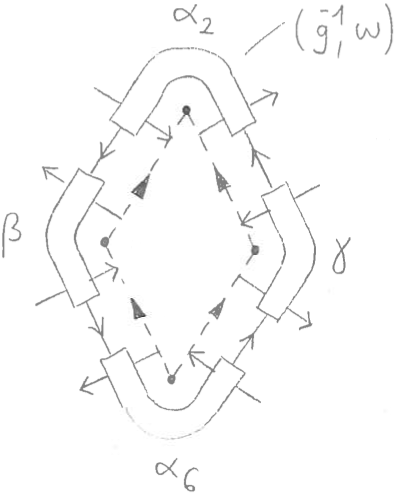}}
           \vcorrect{49bp}{\includegraphics[scale=0.25]{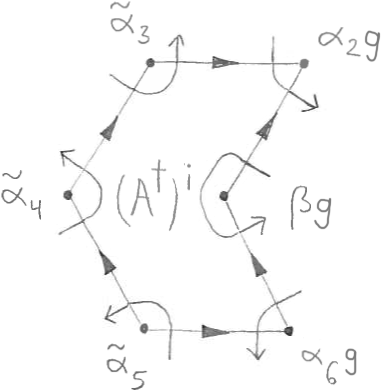}}
           \vcorrect{48bp}{\includegraphics[scale=0.25]{figure68}}
           \displaybreak[0] \\
        & =\frac{1}
                {\abs{G}}
           \sum_{\alpha_l,
                 \beta,
                 \gamma}
           \sum_{g\in
                 G}\mkern3mu
           \vcorrect{62bp}{\includegraphics[scale=0.25]{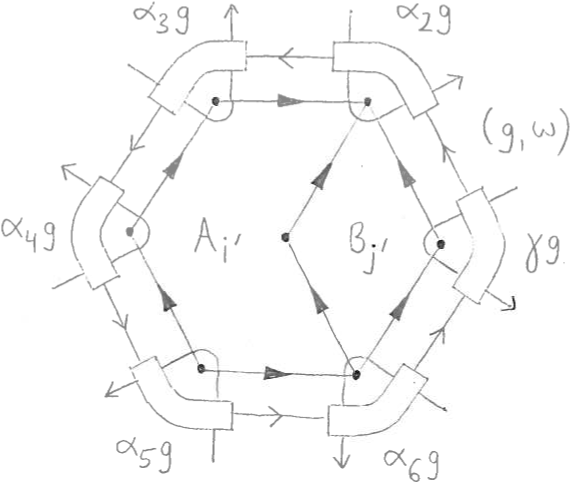}}\mkern12mu
           \vcorrect{49bp}{\includegraphics[scale=0.25]{figure67}} \\
        & \hphantom{={}}
           \vcorrect{62bp}{\includegraphics[scale=0.25]{figure70}}
           \vcorrect{73bp}{\includegraphics[scale=0.25]{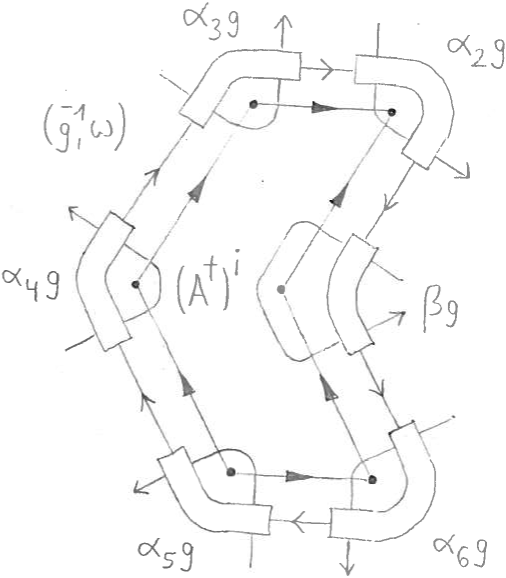}}
           \vcorrect{48bp}{\includegraphics[scale=0.25]{figure68}}\mkern9mu
           \displaybreak[0] \\
        & =\frac{1}
                {\abs{G}}
           \sum_{\alpha_l,
                 \beta}
           \sum_{g\in
                 G}\mkern3mu
           \vcorrect{48bp}{\includegraphics[scale=0.25]{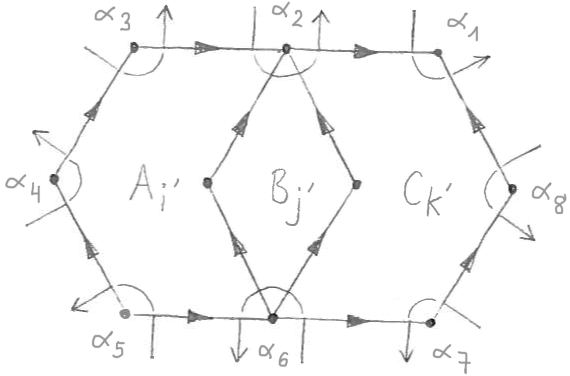}}\mkern9mu
           \vcorrect{72bp}{\includegraphics[scale=0.25]{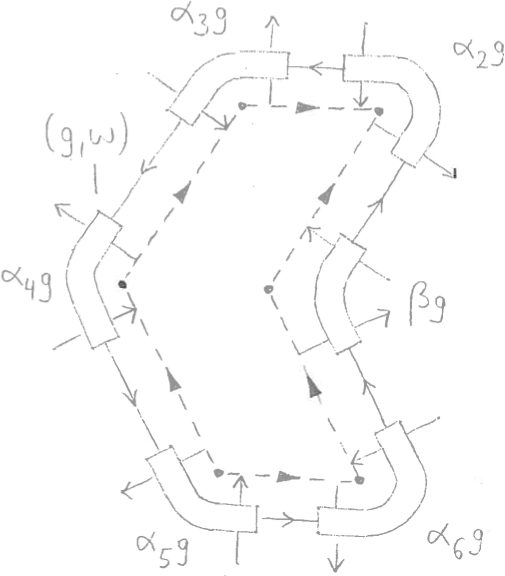}} \\
        & \hphantom{={}}
           \vcorrect{73bp}{\includegraphics[scale=0.25]{figure73}}
           \vcorrect{48bp}{\includegraphics[scale=0.25]{figure68}}
           \displaybreak[0] \\
        & =\sum_{\alpha_l}\mkern3mu
           \vcorrect{48bp}{\includegraphics[scale=0.25]{figure74}}\mkern9mu
           \vcorrect{53bp}{\includegraphics[scale=0.25]{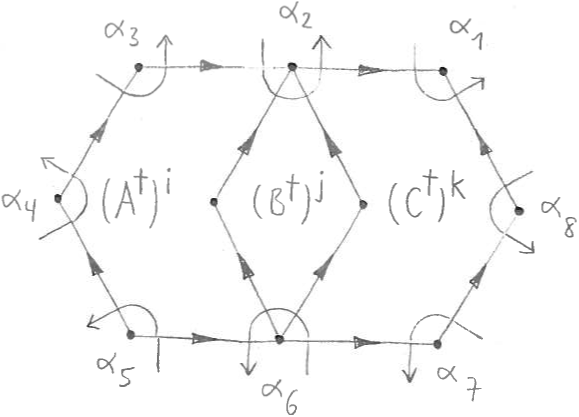}}.
    \end{align*}
    This same argument can readily be applied to
    $\expval{i'\mkern-3mu,j'\mkern-3mu,k'}{\tilde{h}_2\tilde{h}_1}{i,j,k}$, hence we have
    shown that $[\tilde{h}_1,\tilde{h}_2]=[h_1,h_2]=0$.
\end{proof}

\section{Gapped Paths and Almost Twisted Isometry}
\label{sec:almost_twisted_isometry}

Let $\{A,\dots\}$ be a set of compatible $(G,\omega)$-injective tensors.
Can we connect the parent Hamiltonians of
\begin{equation*}
    \ket{\psi(A,
              \dots)}
    =Q\otimes
     \cdots\mkern2mu
     \ket{\psi(\mathcal{P}^\omega,
               \dots)}
\end{equation*}
and $\ket{\psi(\mathcal{P}^\omega,\dots)}$
by quasi-adiabatic evolution, i.e. are they in the same universality class?
Under what circumstances can we remove the positive maps~$Q$?

It is known that this question hinges on the existence of an energy gap in the
thermodynamical limit which is notoriously difficult to prove in two
dimensions. Already for the case of completely trivial virtual symmetry, i.e.
for $(G,\omega)\simeq(\{e\},1)$ which is known as the injective case, there are examples of
\textsc{PEPS} whose parent Hamiltonians are gapless~\cite{Verstraete:2006jk,Hastings:2006gv},
in sharp contrast to injective \textsc{MPS}. As we have shown, parent
Hamiltonians of $(G,\omega)$-isometric \textsc{PEPS} tensors \emph{do} have a gap in the
thermodynamic limit, hence the deformations~$Q$ cannot be removed in general without
changing the nature of the gap. However, it was shown in~\cite{Schuch:2011p3031} that all
positive maps~$Q$ in a neighbourhood of the identity \emph{can} be removed
by quasi-adiabatic evolution.

\begin{definition}
    A $(G,\omega)$-injective tensor (with standard decomposition
    $Q\mathcal{P}^\omega$) is called \emph{almost $(G,\omega)$-isometric}
    if $Q$ is in a neighbourhood of the identity.%
    \footnote{This neighbourhood is defined via the spectrum
        of~$Q$~\cite{Schuch:2011p3031}.}
\end{definition}

\begin{theorem}
    \label{thm:Hamiltonians}
    The universality class of parent Hamiltonians
    associated with almost $(G,\omega)$-isometric
    \textsc{PEPS} only depends on the cohomology
    class of~$\omega$.
\end{theorem}
\begin{proof}
    Let us first look at the \textsc{PEPS} themselves.
    Let $\{A,\dots\}$ and $\{A',\dots\}$ be sets of almost $(G,\omega)$- and $(G,\omega')$-isometric
    tensors respectively. Assume that both sets agree on the
    virtual boundaries of their tensors.
    Since all~$Q'$ are in a neighbourhood of the identity we can
    reduce each local tensor~$A'$ to the standard form~$\mathcal{P}^{\omega'}$
    by evolving the state~$\ket{\psi(A',\dots)}$ quasi-adiabatically.
    If $\omega'\sim\omega$ we may subsequently turn
    $\mathcal{P}^{\omega'}$ into $\mathcal{P}^\omega$ by Proposition~\ref{prop:equivalence_isometric}.
    Finally, we deform each local tensor~$\mathcal{P}^\omega$ by~$Q$ to
    obtain the state~$\ket{\psi(A,\dots)}$.%
    \footnote{Alternatively, we have
        $A'=Q'\mathcal{P}^{\omega'}\simeq Q_\phi'\mathcal{P}^\omega V_\phi\simeq Q_\phi'\mathcal{P}^\omega\eqqcolon A$.
        Here we set $Q_\phi'\coloneqq V_\phi Q' V_\phi^\dagger$, the first
        equivalence is due to a local unitary at the physical level and the second
        one due to contraction at the virtual level. Note that
        $Q_\phi'\neq Q'$ since $Q'$ and $V_\phi$ do \emph{not} commute in
        general, however, $Q_\phi'$ and $Q'$ have the same spectrum. Thus
        $A'$ is almost $(G,\omega')$-isometric if and only if $A$
        is almost $(G,\omega)$-isometric.}
    
    We have just shown that
    the states $\ket{\psi(A,\dots)}$
    and $\ket{\psi(A',\dots)}$ are equivalent under quasi-adiabatic evolution.
    This immediately carries over to their respective parent Hamiltonians.
\end{proof}

Clearly, the Hamiltonians of Theorem~\ref{thm:Hamiltonians}
are gapped (because they are equivalent to parent Hamiltonians at the $(G,\omega)$-isometric
point), frustration-free and their
interaction terms do not commute (because of the deformations~$Q$). The universality
class is determined by the $(G,\omega)$-isometric point and the cohomology class
of~$\omega$, i.e. by the new standard form~$\mathcal{P}^\omega$. This point corresponds to a coarse-grained time slice of the \textsc{Dijkgraaf-Witten}
\textsc{TQFT}, or equivalently, to a coarse-grained twisted quantum double model.

\section{Discussion and Outlook}
\label{sec:discussion}

In this article we presented a new standard form for \textsc{PEPS},
which is based on the twisted action of a finite group on the
virtual boundary of a tensor.
Each group element~$g$ acts via an \textsc{MPO}~$V^\omega(g)$. As long
as the twist~$\omega$ is nontrivial \emph{each} such \textsc{MPO}
is entangled, i.e. it has a virtual bond dimension $D=\abs{G}>1$. While
$V^\omega$ is still a \emph{linear} representation on the tensor product of
virtual spaces (along the virtual boundary), the twist~$\omega$
modifies the associativity between those tensor factors nontrivially,
which affects the symmetric subspace. The nontrivial interplay between this
twisted symmetric subspace and the different (periodic) boundary conditions on a torus is
what can accommodate for the many new universality classes compared to the untwisted case~\cite{Schuch:2010p2806}.

In fact, the symmetric subspace is the fundamental object. The twisted
group action is merely a tool to consistently describe a family of virtual subspaces
which grow at the correct asymptotic rate as the tensor network is
grown (which is nothing but the intersection condition of Theorem~\ref{thm:intersection}).
For example, the (possibly twisted) action of certain Hopf
algebras~\cite{Buerschaper:2013ky} also
leads to a consistent family of subspaces and describes yet another kind
of intrinsic topological order in two dimensions. In general,
it may be more convenient to define the family of virtual subspaces
in terms of projections with certain relations rather than by an explicit
algebraic action. Indeed, it should be possible to derive a suitable
family of virtual projections for any \textsc{Turaev-Viro}
invariant~\cite{Turaev:1992fp,Barrett:1996dn} by generalizing the correspondence between
$(G,\omega)$-isometric tensor networks and
the $(2+1)$-dimensional \textsc{Dijkgraaf-Witten} partition function
outlined in this article. The flatness condition for triangles will have to
be replaced by the fusion rule of a spherical fusion category, and since
two edges of a flat triangle may no longer determine the third one uniquely,
additional indices will be required for the \textsc{MPO} describing
the projections. This virtual \enquote{symmetry} will turn out to be the one
arising in (tensor network representations of) the
ground states~\cite{Gu:2009p1885,Buerschaper:2009p1923}
of general \textsc{Levin-Wen} models~\cite{Levin:2004p117,KirillovJr:2011wq}.
As such it will provide a standard form for \textsc{PEPS} which further
generalizes the one introduced in this article.

We also presented a family of gapped, frustration-free Hamiltonians with
non-commuting interaction terms whose universality class is determined by
the new standard form~$\mathcal{P}^\omega$ for \textsc{PEPS}. Previously,
this class has mainly been studied using two approaches: 
quantum field theory~\cite{Dijkgraaf:1990p2241,deWildPropitius:1997p3081,Hung:2012vw}
and exactly solved models~\cite{Hu:mIFulyz6,Mesaros:2013es}. In the first
approach it is not always clear whether a particular field theory
can actually be implemented
by a microscopic Hamiltonian.
The second approach yields comparatively simple Hamiltonians which consist of
commuting projections. Our results address both these shortcomings and may thus
provide a promising route for further progress in the classification
of 2D gapped quantum spin systems.

It is a very interesting question
how this classification of $(2+1)$-dimensional phases via virtual \enquote{symmetry} needs to be
extended to incorporate (physical) symmetry protection or higher dimensions.

\section*{Acknowledgements}

The author acknowledges stimulating discussions with
Ignacio Cirac,
Zheng-Cheng Gu,
Alexander Kirillov Jr.,
David Pérez-García,
Norbert Schuch and
Frank Verstraete.
He would also like to thank Guifré Vidal for valuable comments
on an earlier version of this manuscript.

This research was supported in part by Perimeter Institute for Theoretical Physics.
Research at Perimeter Institute is supported by the Government of
Canada through Industry Canada and by the Province of Ontario through
the Ministry of Research and Innovation.

\begin{raggedright}
    \printbibliography
\end{raggedright}

\appendix

\section{Group cohomology}

\subsection{3-Cocycles}

A 3-cocycle~$\omega$ on a finite group~$G$ is a function $G\times G\times G\to\mathrm{U}(1)$
which satisfies
\begin{equation}
    \label{eq:3-cocycle_definition}
    \omega(a,
           b,
           c)\mkern2mu
    \omega(a,
           b
           c,
           d)\mkern2mu
    \omega(b,
           c,
           d)
    =\omega(a
            b,
            c,
            d)\mkern2mu
     \omega(a,
            b,
            c
            d)
\end{equation}
for all $a$, $b$, $c$, $d\in G$.

Recall that a 3-cocycle can be obtained from any 2-cochain~$\phi$ via
\begin{equation}
    \label{eq:3-coboundary_definition}
    (\mathrm{d}
     \phi)(a,
           b,
           c)
    \coloneqq
     \frac{\phi(a,
                b
                c)\mkern2mu
           \phi(b,
                c)}
          {\phi(a,
                b)\mkern2mu
           \phi(a
                b,
                c)}.
\end{equation}
This is in fact called a 3-coboundary. We regard two 3-cocycles~$\omega$
and~$\omega'$ that only differ by
a 3-coboundary~$\mathrm{d}\phi$ as equivalent, i.e. $\omega$ and~$\omega'$ are in the same cohomology class.
The set of such cohomology classes forms an Abelian group~$H^3\bigl(G,\mathrm{U}(1)\bigr)$.

\bigskip
A 3-cocycle~$\tilde{\omega}$ is called \emph{normalized} if
$\tilde{\omega}(1,b,c)=\tilde{\omega}(a,1,c)=\tilde{\omega}(a,b,1)=1$
for all $a$, $b$, $c\in G$. Every 3-cocycle~$\omega$ is equivalent
to a normalized 3-cocycle~$\tilde{\omega}$ for some~$\phi$. In particular, this means:
\begin{align}
    \omega(1,
           b,
           c)
    & =\frac{\phi(1,
                  b
                  c)}
            {\phi(1,
                  b)},
       \label{eq:reduction1} \\
    \omega(a,
           1,
           c)
    & =\frac{\phi(1,
                  c)}
            {\phi(a,
                  1)},
       \label{eq:reduction2} \\
    \omega(a,
           b,
           1)
    & =\frac{\phi(b,
                  1)}
            {\phi(a
                  b,
                  1)}.
       \label{eq:reduction3}
\end{align}

\subsection{Twisted 2-Cocycles}

Recall that a twisted 2-cocycle~$c$ on a finite group~$G$ is a function
$G\times G\times G\to\mathrm{U}(1)$ which satisfies
\begin{equation}
    \label{eq:twisted_2-cocycle}
    c_g(x,
        y)\mkern2mu
    c_g(x
        y,
        z)
    =c_g(x,
         y
         z)\mkern2mu
     c_{x^{-1}
        g
        x}(y,
           z)
\end{equation}
for all $g$, $x$, $y$, $z\in G$.

Let us record some useful properties of twisted 2-cocycles.

\begin{lemma}
    \label{lem:c_pair_conjugation}
    Let $g,h\in G$ with $hg=gh$.
    
    Then for all $t\in G$ and $x\in Z(g^t\mkern-2mu,h^t)$:
    \begin{equation}
        \frac{c_{g^t}(h^t\mkern-2mu,
                      x)}
             {c_{g^t}(x,
                      h^t)}
        =\frac{c_g(h,
                   \prescript{t}{}{\mkern-1mu x})}
              {c_g(\prescript{t}{}{\mkern-1mu x},
                   h)}.
    \end{equation}
    
    In particular, $(g,h)$ is a $c$-regular pair if and only
    if $(g^t\mkern-2mu,h^t)$ is a $c$-regular pair.
\end{lemma}
\begin{proof}
    Note first that $x\in Z(g^t\mkern-2mu,h^t)$ is equivalent to
    $\prescript{t}{}{\mkern-1mu x}\in Z(g,h)$.
    Repeatedly using~\eqref{eq:twisted_2-cocycle} we then obtain
    \begin{align*}
        \frac{c_{g^t}(h^t\mkern-2mu,
                      x)}
             {c_{g^t}(x,
                      h^t)}
        & =\frac{c_g(h
                     t^{-1}\mkern-3mu,
                     x)\mkern2mu
                 c_g(t^{-1}\mkern-3mu,
                     h^t)}
                {c_g(t^{-1}\mkern-3mu,
                     h^t
                     x)}\mkern2mu
           \frac{c_g(t^{-1}\mkern-3mu,
                     x
                     h^t)}
                {c_g(t^{-1}
                     x,
                     h^t)\mkern2mu
                 c_g(t^{-1}\mkern-3mu,
                     x)} \\
        & =\frac{c_g(h
                     t^{-1}\mkern-3mu,
                     x)\mkern2mu
                 c_g(t^{-1}\mkern-3mu,
                     h^t)}
                {c_g(t^{-1}
                     x,
                     h^t)\mkern2mu
                 c_g(t^{-1}\mkern-3mu,
                     x)} \\
        & =\frac{c_g(t^{-1}\mkern-3mu,
                     x)\mkern2mu
                 c_g(h,
                     t^{-1}
                     x)}
                {c_g(h,
                     t^{-1})}\mkern2mu
           \frac{c_g(t^{-1}\mkern-3mu,
                     h^t)}
                {c_g(t^{-1}
                     x,
                     h^t)\mkern2mu
                 c_g(t^{-1}\mkern-3mu,
                     x)} \\
        & =\frac{c_g(h,
                     \prescript{t}{}{\mkern-1mu x}
                     t^{-1})}
                {c_g(h,
                     t^{-1})}\mkern2mu
           \frac{c_g(t^{-1}\mkern-3mu,
                     h^t)}
                {c_g(\prescript{t}{}{\mkern-1mu x}
                     t^{-1}\mkern-3mu,
                     h^t)} \\
        & =\frac{c_g(h
                     \prescript{t}{}{\mkern-1mu x},
                     t^{-1})\mkern2mu
                 c_g(h,
                     \prescript{t}{}{\mkern-1mu x})}
                {c_g(\prescript{t}{}{\mkern-1mu x},
                     t^{-1})}\mkern2mu
           \frac{c_g(\prescript{t}{}{\mkern-1mu x},
                     t^{-1})}
                {c_g(t^{-1}\mkern-3mu,
                     h^t)\mkern2mu
                 c_g(\prescript{t}{}{\mkern-1mu x},
                     h
                     t^{-1})}\mkern2mu
           \frac{c_g(t^{-1}\mkern-3mu,
                     h^t)}
                {c_g(h,
                     t^{-1})} \\
        & =\frac{c_g(h
                     \prescript{t}{}{\mkern-1mu x},
                     t^{-1})\mkern2mu
                 c_g(h,
                     \prescript{t}{}{\mkern-1mu x})}
                {c_g(\prescript{t}{}{\mkern-1mu x},
                     h
                     t^{-1})\mkern2mu
                 c_g(h,
                     t^{-1})} \\
        & =\frac{c_g(\prescript{t}{}{\mkern-1mu x}
                     h,
                     t^{-1})\mkern2mu
                 c_g(h,
                     \prescript{t}{}{\mkern-1mu x})}
                {c_g(\prescript{t}{}{\mkern-1mu x},
                     h
                     t^{-1})\mkern2mu
                 c_g(h,
                     t^{-1})} \\
        & =\frac{c_g(h,
                     t^{-1})\mkern2mu
                 c_g(\prescript{t}{}{\mkern-1mu x},
                     h
                     t^{-1})}
                {c_g(\prescript{t}{}{\mkern-1mu x},
                     h)}\mkern2mu
           \frac{c_g(h,
                     \prescript{t}{}{\mkern-1mu x})}
                {c_g(\prescript{t}{}{\mkern-1mu x},
                     h
                     t^{-1})\mkern2mu
                 c_g(h,
                     t^{-1})} \\
        & =\frac{c_g(h,
                     \prescript{t}{}{\mkern-1mu x})}
                {c_g(\prescript{t}{}{\mkern-1mu x},
                     h)}.
    \end{align*}
\end{proof}

\begin{lemma}
    \label{lem:twisted_cohomology}
    Let $\epsilon$ a function $G\times G\to\mathrm{U}(1)$ and
    \begin{equation}
        c_g'(x,
             y)
        =(c_g\mkern2mu
          \mathrm{d}
          \epsilon_g)(x,
                      y)
        \coloneqq
         c_g(x,
             y)\mkern2mu
         \frac{\epsilon_g(x)\mkern2mu
               \epsilon_{x^{-1}
                         g
                         x}(y)}
              {\epsilon_g(x
                          y)}.
    \end{equation}
    
    Then
    \begin{enumerate}
        \item $c'$ is a twisted 2-cocycle.
        \item A pair~$(g,h)$ is $c$-regular if and only if it is $c'$-regular.
    \end{enumerate}
\end{lemma}
\begin{proof}
    The first claim follows directly from~\eqref{eq:twisted_2-cocycle}.
    
    The second claim follows from
    \begin{equation*}
        c_g'(h,
             x)
        =c_g(h,
             x)\mkern2mu
         \frac{\epsilon_g(h)\mkern2mu
               \epsilon_{h^{-1}
                         g
                         h}(x)}
              {\epsilon_g(h
                          x)}
        =c_g(x,
             h)\mkern2mu
         \frac{\epsilon_g(x)\mkern2mu
               \epsilon_{x^{-1}
                         g
                         x}(h)}
              {\epsilon_g(x
                          h)}
        =c_g'(x,
              h)
    \end{equation*}
    for $h\in Z(G)$ and all $x\in Z(g,h)$.
\end{proof}

\subsection{Twisted 2-Cocycles from 3-Cocycles}

We turn to the particular twisted 2-cocycles~\eqref{eq:c_omega_definition}
which are defined in terms of (ordinary) 3-cocycles.
 
\begin{lemma}
    \label{lem:c_special}
    Let $\omega$ be a 3-cocycle on a finite group.
    Then
    \begin{enumerate}
        \item $c^\omega$ is a twisted 2-cocycle.
        \item If $\omega$ is normalized, then so is~$c^\omega$.
        \item If $\omega'=\omega\mkern2mu\mathrm{d}\phi$ for any 2-cochain~$\phi$,
            then $c^{\omega'}=c^\omega\mkern2mu\mathrm{d}\epsilon$
            with
            \begin{equation}
                \epsilon_g(x)
                =\frac{\phi(x,
                            \prescript{x\mkern-2mu}{}{g})}
                      {\phi(g,
                            x)}.
            \end{equation}
    \end{enumerate}
\end{lemma}
\begin{proof}
    The first statement follows from~\eqref{eq:3-cocycle_definition}
    and~\eqref{eq:twisted_2-cocycle}, the second statement
    is obvious.
    The third claim follows from~\eqref{eq:3-coboundary_definition}
    by a simple calculation.
\end{proof}

\begin{lemma}
    \label{lem:c_omega_regularity}
    One has
    \begin{enumerate}
        \item A pair~$(g,h)$ is $c^\omega$-regular if and only if
            $(h,g)$ is.
        \item Every pair~$(g,g)$ is $c^\omega$-regular.
    \end{enumerate}
\end{lemma}
\begin{proof}
    By definition, a $c^\omega$-regular pair~$(g,h)$ is
    equivalent to requiring $hg=gh$ and
    \begin{equation}
        \label{eq:c_omega}
        \frac{c_g^\omega(h,
                         x)}
             {c_g^\omega(x,
                         h)}
        =\frac{\omega(g,
                      h,
                      x)\mkern2mu
               \omega(h,
                      x,
                      g)\mkern2mu
               \omega(x,
                      g,
                      h)}
              {\omega(h,
                      g,
                      x)\mkern2mu
               \omega(g,
                      x,
                      h)\mkern2mu
               \omega(x,
                      h,
                      g)}
        =1       
    \end{equation}
    for all $x\in Z(g,h)$.
    
    1. It is easily seen from~\eqref{eq:c_omega} that
    \begin{equation*}
        \frac{c_g^\omega(h,
                         x)}
             {c_g^\omega(x,
                         h)}
        =\biggl(\frac{c_h^\omega(g,
                                 x)}
                     {c_h^\omega(x,
                                 g)}\biggr)^{-1},
    \end{equation*}
    and the claim follows directly.
    
    2. Obvious from~\eqref{eq:c_omega}.
\end{proof}

\subsection{Properties of~$\eta$}

We need to study the properties of~$\eta_g(x,y)$ as defined
by~\eqref{eq:eta_definition} in two situations: first for
the general case $y\in G$ and then for the special case
$y\in Z(g,x)$. In the following we always assume that $x\in Z(g)$.

\begin{lemma}
    \label{lem:eta_property1}
    Let $y\in G$. Then
    \begin{equation}
        \eta_{g^t}(x^t\mkern-2mu,
                   y
                   t^{-1})
        =\frac{\eta_g(x,
                      y)}
              {\eta_g(x,
                      t)}.
    \end{equation}
    holds for all $t\in G$.
\end{lemma}
\begin{proof}
    It will be most convenient to prove this directly from
    the twisted 2-cocycle property of~$c^\omega$ without
    invoking~\eqref{eq:c_omega_definition}. Hence the following will in fact hold
    for arbitrary twisted 2-cocycles~$c$:
    \begin{align*}
        \eta_{g^t}(x^t\mkern-2mu,
                   y
                   \bar{t})
        & =\frac{c_{g^t}(t
                         \bar{y},
                         x^y)}
                {c_{g^t}(x^t\mkern-2mu,
                         t
                         \bar{y})} \\
        & =\frac{c_g(\bar{t},
                     t
                     \bar{y})\mkern2mu
                 c_g(\bar{y},
                     x^y)}
                {c_g(\bar{t},
                     t
                     x
                     \bar{y})}
           \frac{c_g(\bar{t},
                     t
                     x
                     \bar{y})}
                {c_g(\bar{t},
                     x^t)\mkern2mu
                 c_g(x
                     \bar{t},
                     t
                     \bar{y})} \\
        & =\frac{c_g(\bar{t},
                     t
                     \bar{y})\mkern2mu
                 c_g(\bar{y},
                     x^y)}
                {c_g(\bar{t},
                     x^t)}\mkern2mu
           \frac{1}
                {c_g(x
                     \bar{t},
                     t
                     \bar{y})} \\
        & =\frac{c_g(\bar{t},
                     t
                     \bar{y})\mkern2mu
                 c_g(\bar{y},
                     x^y)}
                {c_g(\bar{t},
                     x^t)}\mkern2mu
           \frac{c_g(x,
                     \bar{t})}
                {c_g(x,
                     \bar{y})\mkern2mu
                 c_{\bar{x}
                    g
                    x}(\bar{t},
                       t
                       \bar{y})} \\
        & =\frac{c_g(\bar{y},
                     x^y)}
                {c_g(x,
                     \bar{y})}\mkern2mu
           \frac{c_g(x,
                     \bar{t})}
                {c_g(\bar{t},
                     x^t)} \\
        & =\frac{\eta_g(x,
                        y)}
                {\eta_g(x,
                        t)}.
    \end{align*}
    Note that we repeatedly applied~\eqref{eq:twisted_2-cocycle} in the above derivation.
\end{proof}

Now we turn to the special case.
\begin{lemma}
    Let $y\in Z(g,x)$.
    Then
    \begin{align}
        \eta_x(g,
               y)
        & =\eta_g(x,
                  y)^{-1}, \\
        \eta_g(y,
               x)
        & =\eta_g(x^{-1}\mkern-3mu,
                  y^{-1})^{-1}.
    \end{align}
\end{lemma}
\begin{proof}
    Note first that the assumption is equivalent to requiring
    $xg=gx\wedge yg=gy\wedge yx=xy$,
    in other words, all arguments of~$\eta$ mutually commute.
    Next we see that~\eqref{eq:eta_definition} simplifies to
    \begin{equation}
        \label{eq:eta_simplified}
        \eta_g(x,
               y)
        =\frac{c_g^\omega(y^{-1}\mkern-3mu,
                          x)}
              {c_g^\omega(x,
                          y^{-1})}
        =\frac{\omega(g,
                      y^{-1}\mkern-3mu,
                      x)\mkern2mu
               \omega(y^{-1}\mkern-3mu,
                      x\mkern-2mu,
                      g)\mkern2mu
               \omega(x,
                      g,
                      y^{-1})}
              {\omega(y^{-1}\mkern-3mu,
                      g,
                      x)\mkern2mu
               \omega(g,
                      x,
                      y^{-1})\mkern2mu
               \omega(x,
                      y^{-1}\mkern-3mu,
                      g)}.
    \end{equation}
    from where one can immediately verify both claims.
\end{proof}

\begin{lemma}
    \label{lem:eta_neutral}
    \begin{equation}
        \eta_g(x,
               e)
        =1.
    \end{equation}
\end{lemma}
\begin{proof}
    It is immediate from~\eqref{eq:reduction1}, \eqref{eq:reduction2}
    and~\eqref{eq:reduction3}
    that~\eqref{eq:eta_definition} simplifies to
    \begin{equation*}
        \eta_g(x,
               e)
        =\frac{\omega(g,
                      e,
                      x)\mkern2mu
               \omega(e,
                      x,
                      g)\mkern2mu
               \omega(x,
                      g,
                      e)}
              {\omega(g,
                      x,
                      e)\mkern2mu
               \omega(e,
                      g,
                      x)\mkern2mu
               \omega(x,
                      e,
                      g)}
        =1.
    \end{equation*}
\end{proof}

\end{document}